\documentclass{article}

\usepackage[final]{neurips_2022}

\usepackage[utf8]{inputenc} 
\usepackage[T1]{fontenc}    
\usepackage{hyperref}       
\usepackage{url}            
\usepackage{booktabs}       
\usepackage{amsfonts}       
\usepackage{nicefrac}       
\usepackage{microtype}      
\usepackage{xcolor}         

\usepackage{amssymb,amsmath,amsthm,enumerate,threeparttable,bm,graphicx,subfigure}
\usepackage{algorithm, algorithmic}
\usepackage{booktabs}
\usepackage{multirow}
\usepackage{lipsum}
\usepackage{enumitem}

\newtheorem{lemma}{Lemma}
\newtheorem{theorem}{Theorem}
\newtheorem{definition}{Definition}
\newtheorem{proposition}{Proposition}

\long\def\comment#1{}
\usepackage{setspace}
\usepackage{amsmath}
\usepackage{wrapfig}
\newcommand{\etal}{\textit{et al}. }

\def\eg{$e.g.$}

\newcommand{\subsec}{\noindent \textbf}

\title{Untargeted Backdoor Watermark: Towards Harmless and Stealthy Dataset Copyright Protection}

\author{Yiming Li$^{1,}$\thanks{The first two authors contributed equally to this work. Correspondence to: Yang Bai and Shu-Tao Xia.} , Yang Bai$^{2, \ast}$, Yong Jiang$^{1}$, Yong Yang$^{3}$, Shu-Tao Xia$^{1}$, Bo Li$^{4}$\\
$^{1}$Tsinghua Shenzhen International Graduate School, Tsinghua University, China\\
$^{2}$Tencent Security Zhuque Lab, China\\
$^{3}$Tencent Security Platform Department, China\\
$^{4}$The Department of Computer Science, University of Illinois at Urbana-Champaign, USA\\
\texttt{li-ym18@mails.tinghua.edu.cn}; \texttt{\{mavisbai,coolcyang\}@tencent.com};\\
\texttt{\{jiangy,xiast\}@sz.tsinghua.edu.cn}; \texttt{lbo@illinois.edu}
}

\begin{document}

\maketitle

\begin{abstract}
Deep neural networks (DNNs) have demonstrated their superiority in practice. Arguably, the rapid development of DNNs is largely benefited from high-quality (open-sourced) datasets, based on which researchers and developers can easily evaluate and improve their learning methods. Since the data collection is usually time-consuming or even expensive, how to protect their copyrights is of great significance and worth further exploration. In this paper, we revisit dataset ownership verification. We find that existing verification methods introduced new security risks in DNNs trained on the protected dataset, due to the targeted nature of poison-only backdoor watermarks. To alleviate this problem, in this work, we explore the untargeted backdoor watermarking scheme, where the abnormal model behaviors are not deterministic. Specifically, we introduce two dispersibilities and prove their correlation, based on which we design the untargeted backdoor watermark under both poisoned-label and clean-label settings. We also discuss how to use the proposed untargeted backdoor watermark for dataset ownership verification. Experiments on benchmark datasets verify the effectiveness of our methods and their resistance to existing backdoor defenses. Our codes are available at \url{https://github.com/THUYimingLi/Untargeted_Backdoor_Watermark}.
\end{abstract}

\section{Introduction}
Deep neural networks (DNNs) have been widely and successfully deployed in many applications
, for their effectiveness and efficiency. Arguably, the existence of high-quality open-sourced datasets ($e.g.$, CIFAR-10 \cite{krizhevsky2009learning} and ImageNet \cite{deng2009imagenet}) is one of the key factors for the prosperity of DNNs. Researchers and developers can easily evaluate and improve their methods based on them. However, these datasets may probably be used for commercial purposes without authorization rather than only the educational or academic goals, due to their high accessibility. 

Currently, there were some classical methods for data protection, including encryption, data watermarking, and defenses against data leakage. However, these methods cannot be used to protect the copyrights of open-sourced datasets, since they either hinder the dataset accessibility or functionality ($e.g.$, encryption), require manipulating the training process ($e.g.$, differential privacy), or even have no effect in this case. To the best of our knowledge, there is only one method \cite{li2020open,li2022black} designed for protecting open-sourced datasets. Specifically, it first adopted poison-only backdoor attacks \cite{li2022backdoor} to watermark the unprotected dataset and then conducted ownership verification by verifying whether the suspicious model has specific targeted backdoor behaviors (as shown in Figure \ref{fig:intro_a}).

\begin{figure}[ht]
\centering
\vspace{-2em}
\begin{minipage}[t]{0.45\linewidth}
\centering
\includegraphics[width=\textwidth]{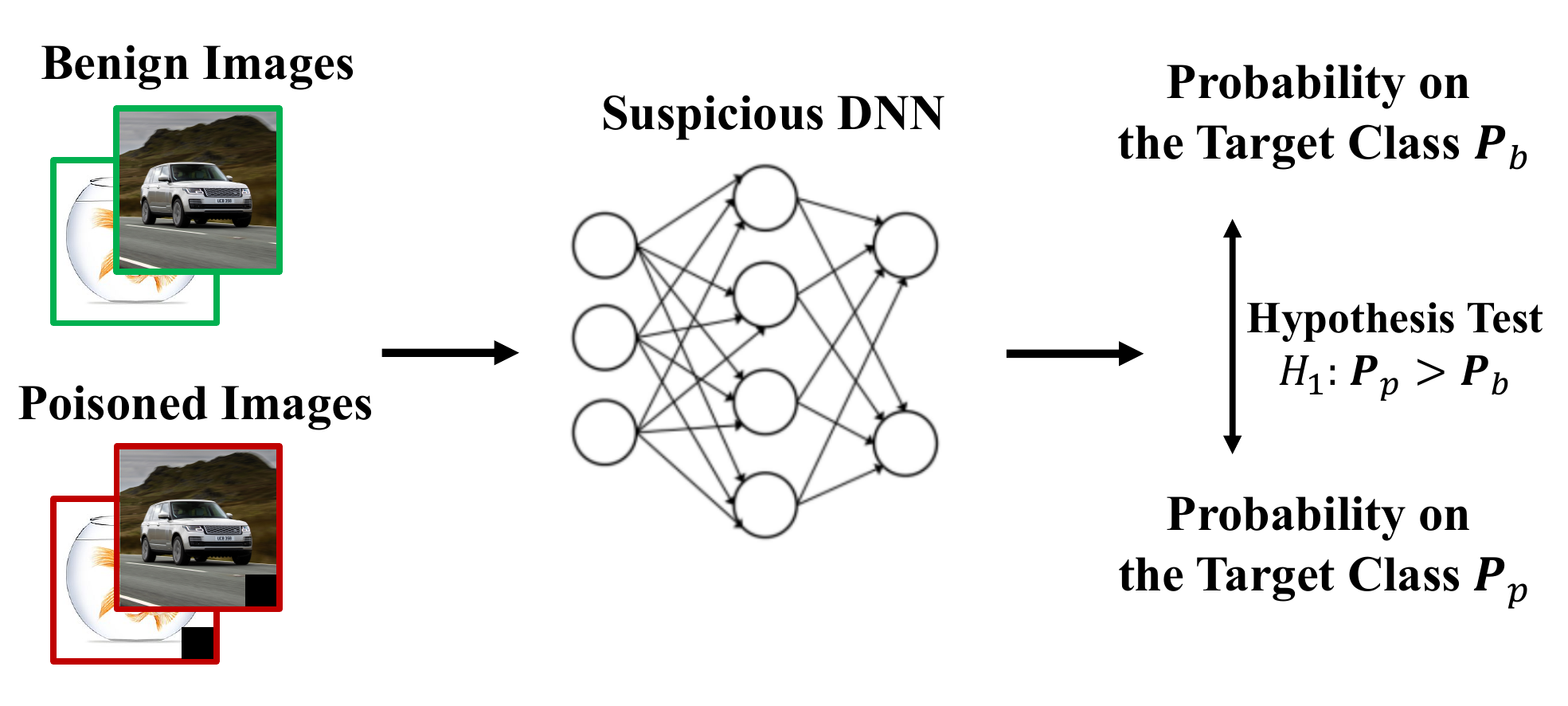}
\vspace{-1.3em}
\caption{The verification process of BEDW.}
\label{fig:intro_a}
\end{minipage}\hspace{1.5em}
\begin{minipage}[t]{0.45\linewidth}
\centering
\includegraphics[width=\textwidth]{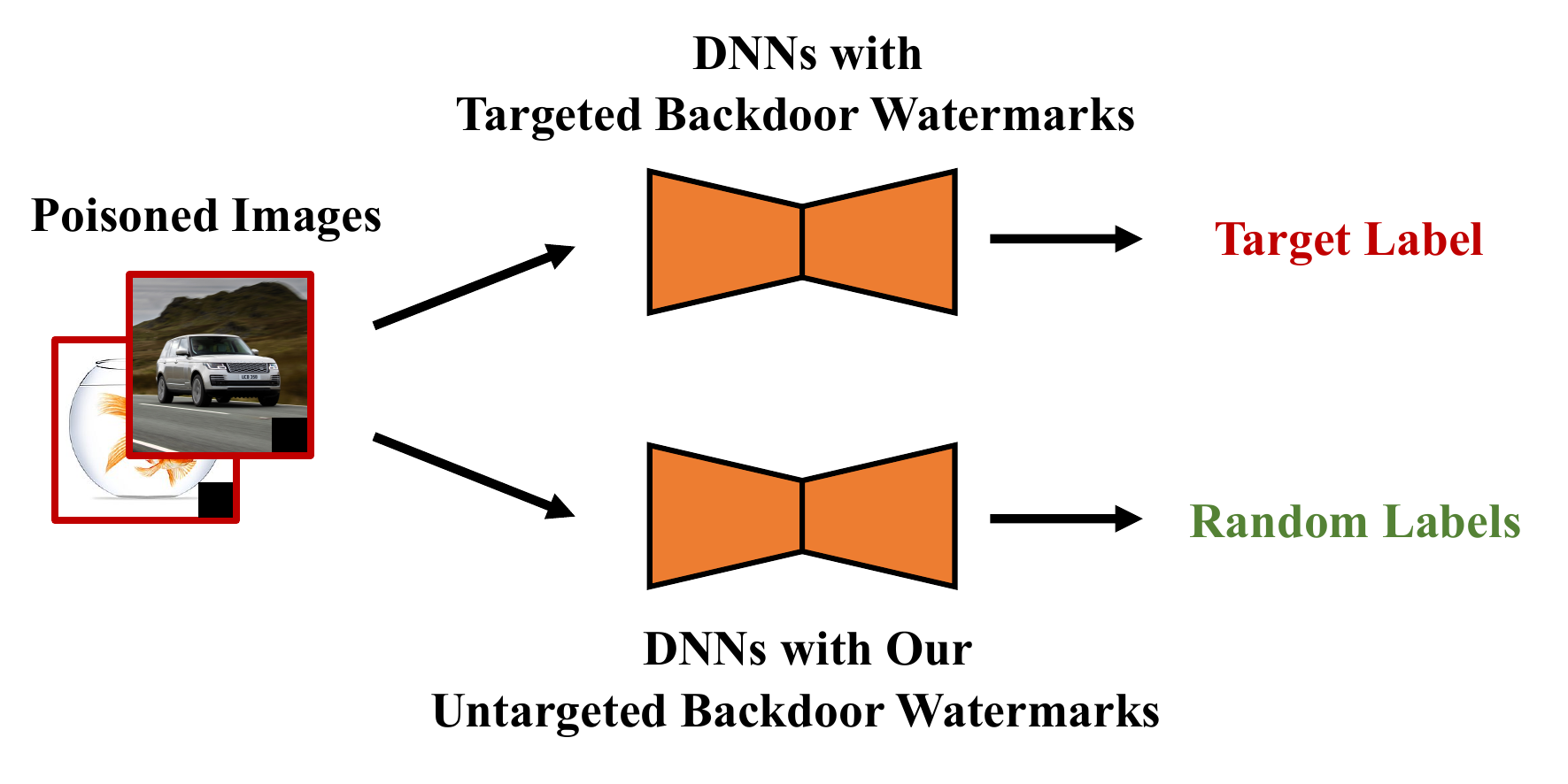}
\vspace{-2em}
\caption{The inference process of DNNs with different types of backdoor watermarks.}
\label{fig:intro_b}
\end{minipage}
\vspace{-1em}
\end{figure}

In this paper, we revisit dataset ownership verification. We argue that BEDW introduced new threatening security risks in DNNs trained on the protected datasets, due to the targeted manner of existing backdoor watermarks. Specifically, the adversaries can exploit the embedded hidden backdoors to maliciously and deterministically manipulate model predictions (as shown in Figure \ref{fig:intro_b}). Based on this understanding, we explore how to design the untargeted backdoor watermark (UBW) and how to use it for harmless and stealthy dataset ownership verification. Specifically, we first introduce two dispersibilities, including averaged sample-wise and averaged class-wise dispersibility, and prove their correlation. Based on them, we propose a simple yet effective heuristic method for UBW with poisoned labels ($i.e.$, UBW-P) and the UBW with clean labels ($i.e.$, UBW-C) based on bi-level optimization. The UBW-P is more effective while the UBW-C is more stealthy. We also design a UBW-based dataset ownership verification, based on the pairwise T-test \cite{hogg2005introduction} at the end.

The main contributions of this paper are four-fold:
\textbf{1)} We reveal the limitations of existing methods in protecting the copyrights of open-sourced datasets; \textbf{2)} We explore the untargeted backdoor watermark (UBW) paradigm under both poisoned-label and clean-label settings; \textbf{3)} We further discuss how to use our UBW for harmless and stealthy dataset ownership verification; \textbf{4)} Extensive experiments on benchmark datasets verify the effectiveness of our method. 

\section{Related Work}
In this paper, we focus on the backdoor watermarks in image classification. The watermarks in other tasks ($e.g.$, \cite{xiang2021backdoor,carlini2022poisoning,li2022few}) and their dataset protection are out of the scope of this paper.

\subsection{Data Protection}
\label{sec:ref_protection}
Data protection aims to prevent unauthorized data usage or protect data privacy, which has always been an important research direction. Currently, encryption, data watermarking, and the defenses against data leakage are the most widespread methods discussed in data protection, as follows:

\subsec{Encryption. }
Currently, encryption is the most widely used data protection method, which intends to encrypt the whole or parts of the protected data \cite{rivest1992md5,boneh2001identity,martins2017survey}. Only authorized users have the secret key to decrypt the encrypted data for further usage. Except for directly preventing unauthorized data usage, there were also some empirical methods focused on encrypting only the sensitive information ($e.g.$, backgrounds or image-label mappings) \cite{xiong2020adgan,li2021visual,cai2021generative}.

\subsec{Data Watermarking. }
This approach was initially used to embed a distinctive watermark into the data to protect its copyright based on ownership verification \cite{swanson1998multimedia,guo2018halftone,abdelnabi2021adversarial}. Recently, data watermarking was also adopted for other applications, such as DeepFake detection \cite{wang2021faketagger} and image steganography \cite{guan2022deepmih}, inspired by its unique properties.

\subsec{Defenses against Data Leakage. }
These methods mainly focus on preventing the leakage of sensitive information ($e.g.$, membership inference \cite{shokri2017membership}, attribute inference \cite{gong2016you}, and deep gradient leakage \cite{zhu2019deep}) during the training process. Among all these methods, differential privacy \cite{dwork2014algorithmic,zhu2021fine,bai2022multinomial} is the most representative one for its good theoretical properties and effectiveness. In general, differential privacy requires to introduce certain randomness via adding noises when training the model.

However, the aforementioned existing methods can not be adopted to prevent open-soured datasets from being unauthorizedly used, since they either hinder dataset functionalities or are not capable in this scenario. To the best of our knowledge, there was only one method \cite{li2020open,li2022black} designed for protecting open-sourced datasets, based on the poison-only targeted backdoor attacks \cite{li2022backdoor}. However, this method will introduce new security threats in the models trained on the protected dataset, which hinders its usage. How to better protect dataset copyrights is still an important open question.

\subsection{Backdoor Attacks}

Backdoor attacks are emerging yet critical threats in the training process of deep neural networks (DNNs), where the adversary intends to embed hidden backdoors into DNNs. The attacked models behave normally in predicting benign samples, whereas the predictions are maliciously changed whenever the adversary-specified trigger patterns appear. Due to this property, they were also used as the watermark techniques for model \cite{adi2018turning,jia2021entangled,li2022defending} and dataset \cite{li2020open,li2022black} ownership verification.

In general, existing backdoor attacks can be divided into three main categories, including \textbf{1)} poison-only attacks \cite{gu2019badnets,turner2019label,nguyen2021wanet}, \textbf{2)} training-controlled attacks \cite{saha2020hidden,zeng2021rethinking,shumailov2021manipulating}, and \textbf{3)} model-modified attacks \cite{rakin2020tbt,wang2022stealthy,qi2022towards}, based on the adversary's capacity levels. In this paper, we only focus on poison-only backdoor attacks, since they are the hardest attack having widespread threat scenarios. Only these attacks can be used to protect open-sourced datasets \cite{li2020open,li2022black}. In particular, based on the label type, existing poison-only attacks can also be separated into two main sub-types, as follows:

\subsec{Poison-only Backdoor Attacks with Poisoned Labels. }
In these attacks, the re-assigned labels of poisoned samples are different from their ground-truth labels. For example, a cat-like poisoned image may be labeled as the dog in the poisoned dataset released by backdoor adversaries. It is currently the most widespread attack paradigm. To the best of our knowledge, BadNets \cite{gu2019badnets} is the first and most representative attack with poisoned labels. Specifically, the BadNets adversary randomly selects certain benign samples from the original benign dataset to generate poisoned samples, based on adding a specific trigger pattern to the images and changing their labels to the pre-defined target label. The adversary will then combine the generated poisoned samples with the remaining benign ones to make the poisoned dataset, which is released to train the attacked models. After that, Chen \emph{et al.} \cite{chen2017targeted} proposed the blended attack, which suggested that the poisoned image should be similar to its benign version to ensure stealthiness. Most recently, a more stealthy and effective attack ($i.e.$, WaNet \cite{nguyen2021wanet}) was proposed, which exploited image warping to design trigger patterns.

\subsec{Poison-only Backdoor Attacks with Clean Labels. } 
Turner \emph{et al.} \cite{turner2019label} proposed the first poison-only backdoor attack with clean labels (\emph{i.e.}, label-consistent attack), where the target label is the same as the ground-truth label of all poisoned samples. They argued that attacks with poisoned labels were not stealthy enough even when the trigger pattern was invisible, since users could still identify the attacks by examining the image-label relation when they caught the poisoned samples. However, this attack is far less effective when the dataset has many classes or high image-resolution ($e.g.$, GTSRB and ImageNet) \cite{souri2022sleeper,huang2022backdoor,li2022backdoor}. Most recently, a more effective attack (\emph{i.e.}, Sleeper Agent) was proposed, which generated trigger patterns by optimization \cite{souri2022sleeper}. Nevertheless, these attacks are still difficult since the `robust features' contained in the poisoned images will hinder the learning of trigger patterns \cite{li2022backdoor}. How to design attacks with clean labels is still left far behind and worth further exploration.

Besides, to the best of our knowledge, all existing backdoor attacks are targeted, $i.e.$, the predictions of poisoned samples are deterministic and known by the adversaries. How to design backdoor attacks in an untargeted manner and its positive applications remain blank and worth further explorations.

\section{Untargeted Backdoor Watermark (UBW)}
\label{sec:UBW}

\subsection{Preliminaries}
\label{sec:prelim}

\subsec{Threat Model. }
In this paper, we focus on poison-only backdoor attacks as the backdoor watermarks in image classification. Specifically, the backdoor adversaries are only allowed to modify some benign samples while having neither the information nor the ability to modify other training components ($e.g.$, training loss, training schedule, and model structure). The generated poisoned samples with remaining unmodified benign ones will be released to victims, who will train their DNNs based on them. In particular, we only consider poison-only backdoor attacks instead of other types of methods (\eg, training-controlled attacks or model-modified attacks) because they require additional adversary capacities and therefore can not be used to protect open-sourced datasets \cite{li2020open,li2022black}.

\subsec{The Main Pipeline of Existing Targeted Backdoor Attacks. }
Let $\mathcal{D} = \{ (\bm{x}_i, y_i) \}_{i=1}^{N}$ denotes the benign training set, where $\bm{x}_i \in \mathcal{X}= \{0,1,\ldots, 255\}^{C\times W \times H}$ is the image, $y_i \in \mathcal{Y} = \{1,\ldots, K\}$ is its label, and $K$ is the number of classes. How to generate the poisoned dataset $\mathcal{D}_{p}$ is the cornerstone of poison-only backdoor attacks. To the best of our knowledge, almost all existing backdoor attacks are \emph{targeted}, where all poisoned samples share the same target label. Specifically, $\mathcal{D}_{p}$ consists of two disjoint parts, including the modified version of a selected subset ($i.e.$, $\mathcal{D}_s$) of $\mathcal{D}$ and remaining benign samples, $i.e.$, 
$\mathcal{D}_{p} =  \mathcal{D}_{m} \cup \mathcal{D}_{b}$, where $y_t$ is an adversary-specified target label, $\mathcal{D}_{b} = \mathcal{D} \backslash \mathcal{D}_{s}$, $\mathcal{D}_{m} = \left\{(\bm{x}', y_t)| \bm{x}' = G(\bm{x};\bm{\theta}), (\bm{x},y) \in \mathcal{D}_s \right\}$, $\gamma \triangleq \frac{|\mathcal{D}_{s}|}{|\mathcal{D}|}$ is the \emph{poisoning rate}, and $G: \mathcal{X} \rightarrow \mathcal{X}$ is an adversary-specified poisoned image generator with parameter $\bm{\theta}$. In particular, poison-only backdoor attacks are mainly characterized by their poison generator $G$. For example, $G(\bm{x}) = (\bm{1}-\bm{\alpha}) \otimes \bm{x} + \bm{\alpha} \otimes \bm{t}$, where $\bm{\alpha} \in [0,1]^{C \times W \times H}$, $\bm{t} \in \mathcal{X}$ is the trigger pattern, and $\otimes$ is the element-wise product in the blended attack \cite{chen2017targeted}; $G(\bm{x}) = \bm{x} + \bm{t}$ in the ISSBA \cite{li2021invisible}. Once the poisoned dataset $\mathcal{D}_{p}$ is generated, it will be released to train DNNs. Accordingly, in the inference process, the attacked model behaves normally on predicting benign samples while its predictions will be maliciously and constantly changed to the target label whenever poisoned images appear.

\subsection{Problem Formulation}
As described in previous sections, DNNs trained on the poisoned dataset will have distinctive behaviors while behaving normally in predicting benign images. As such, the poison-only backdoor attacks can be used to watermark (open-sourced) datasets for their copyright protection. However, this method introduces new security threats in the model since the backdoor adversaries can determine model predictions of malicious samples, due to the targeted nature of existing backdoor watermarks. Motivated by this understanding, we explore untargeted backdoor watermark (UBW) in this paper.

\subsec{Our Watermark's Goals. }
The UBW has three main goals, including \textbf{1)} \emph{effectiveness}, \textbf{2)} \emph{stealthiness}, and \textbf{3)} \emph{dispersibility}. Specifically, the effectiveness requires that the watermarked DNNs will misclassify poisoned images; The stealthiness needs that dataset users can not identify the watermark; The dispersibility (denoted in Definition \ref{def:D_p}) ensures dispersible predictions of poisoned images.

\vspace{0.2em}
\begin{definition}[Averaged Prediction Dispersibility]\label{def:D_p}
Let $\mathcal{D} = \{ (\bm{x}_i, y_i) \}_{i=1}^{N}$ indicates the dataset where $y_i \in \mathcal{Y} = \{1,\ldots, K\}$ and $C: \mathcal{X} \rightarrow \mathcal{Y}$ is a classifier. Let $\bm{P}^{(j)}$ is the probability vector of model predictions on samples having the ground-truth label $j$, where the $i$-th element of $\bm{P}^{(j)}$ is 
\begin{equation}
    P^{(j)}_i \triangleq \frac{\sum_{k=1}^N \mathbb{I}\{C(\bm{x}_k)=i\} \cdot \mathbb{I}\{y_k=j\}}{\sum_{k=1}^N  \mathbb{I}\{y_k=j\}}.
\end{equation}
The averaged prediction dispersibility $D_p$ is defined as

\begin{equation}
    D_p \triangleq \frac{1}{N} \sum_{j=1}^K \sum_{i=1}^N \mathbb{I}\{y_i=j\} \cdot H\left(\bm{P}^{(j)}\right), 
\end{equation}
where $H(\cdot)$ denotes the entropy \cite{kullback1997information}.
\end{definition}

In general, $D_p$ measures how dispersible the predictions of different images having the same label. The larger the $D_p$, the harder that the adversaries can deterministically manipulate the predictions.

\subsection{Untargeted Backdoor Watermark with Poisoned Labels (UBW-P)}

Arguably, the most straightforward strategy to fulfill prediction dispersibility is to make the predictions of poisoned images as the uniform probability vector. Specifically, we propose to randomly `shuffle' the label of poisoned training samples when making the poisoned dataset. This attack is dubbed untargeted backdoor watermark with poisoned labels (UBW-P) in this paper.

Specifically, similar to the existing targeted backdoor watermarks, our UBW-P first randomly select a subset $\mathcal{D}_s$ from the benign dataset $\mathcal{D}$ to make its modified version $\mathcal{D}_m$ by $\mathcal{D}_{m} = \left\{(\bm{x}', y')| \bm{x}' = G(\bm{x};\bm{\theta}), y' \sim [1, \cdots, K], (\bm{x},y) \in \mathcal{D}_s \right\}$, where `$y' \sim [1, \cdots, K]$' denotes sampling $y'$ from the list $[1, \cdots, K]$ with equal probability and $G$ is an adversary-specified poisoned image generator. The modified subset $\mathcal{D}_m$ associated with the remaining benign samples $\mathcal{D} \backslash \mathcal{D}_{s}$ will then be released to train the model $f(\cdot; \bm{w})$ by 
\begin{equation}
    \min_{\bm{w}} \sum_{(\bm{x},y)\in \mathcal{D}_m \cup (\mathcal{D} \backslash \mathcal{D}_{s})} \mathcal{L}(f(\bm{x};\bm{w}),y),
\end{equation}
where $\mathcal{L}$ is the loss function ($e.g.$, cross-entropy \cite{kullback1997information}). 

In the inference process, for any testing sample $(\hat{\bm{x}}, \hat{y}) \notin \mathcal{D}$, the adversary can activate the hidden backdoor contained in attacked DNNs with poisoned image $G(\hat{\bm{x}})$, based on the generator $G$.

\subsection{Untargeted Backdoor Watermark with Clean Labels (UBW-C)}
\label{sec:UBW-C}

As we will demonstrate in Section \ref{sec:exps}, the aforementioned heuristic UBW-P can reach promising results. However, it is not stealthy enough even though the poisoning rate can be small, since UBW-P is still with poisoned labels. Dataset users may identify the watermark by examining the image-label relation when they catch the poisoned samples. In this section, we discuss how to design the untargeted backdoor watermark with clean labels (UBW-C), based on the bi-level optimization \cite{liu2021investigating}.

To formulate UBW-C as a bi-level optimization, we need to optimize the prediction dispersibility. However, it is non-differentiable and therefore cannot be optimized directly. In this paper, we introduce two differentiable surrogate dispersibilities to alleviate this problem, as follows:

\vspace{0.15em}
\begin{definition}[Averaged Sample-wise and Class-wise Dispersibility]
Let $\mathcal{D} = \{ (\bm{x}_i, y_i) \}_{i=1}^{N}$ indicates the dataset where $y_i \in \mathcal{Y} = \{1,\ldots, K\}$, the averaged sample-wise dispersibility of predictions given by the DNN $f(\cdot)$ (over dataset $\mathcal{D}$) is defined as 
\begin{equation}
    D_s \triangleq \frac{1}{N}\sum_{i=1}^{N}  H\left(f(\bm{x}_i)\right), 
\end{equation}
while the class-wise dispersibility is defined as

\begin{equation}
   D_c \triangleq \frac{1}{N} \sum_{j=1}^K   \sum_{i=1}^N \mathbb{I}\{y_i=j\} \cdot H\left(\frac{\sum_{k=1}^N f(\bm{x}_k) \cdot \mathbb{I}\{y_k=j\}}{\sum_{k=1}^N \mathbb{I}\{y_k=j\}}\right).
\end{equation}

\end{definition}

In general, the averaged sample-wise dispersibility describes the average dispersion of predicted probability vectors for all samples, while the averaged class-wise dispersibility depicts the average degree of the dispersion of the average prediction of samples in each class. Maximizing them will have similar effects in optimizing the prediction dispersibility $D_p$.

In particular, the main difference of UBW-C compared with UBW-P and existing targeted backdoor watermarks lies in the generation of the modified subset $\mathcal{D}_m$. Specifically, in UBW-C, \emph{we do not modify the labels of all poisoned samples}, $i.e.$, $\mathcal{D}_{m} = \left\{(\bm{x}', y)| \bm{x}' = G(\bm{x};\bm{\theta}), (\bm{x},y) \in \mathcal{D}_s \right\}$. Before we reach the technical details of our UBW-C, we first present the necessary lemma and theorem.

\begin{lemma}
\label{lemma1}
The averaged class-wise dispersibility is always greater than the averaged sample-wise dispersibility divided by $N$, $i.e.$, 
$D_c > \frac{1}{N} \cdot D_s$.
\end{lemma}

\vspace{0.15em}
\begin{theorem}
\label{thm1}
Let $f(\cdot;\bm{w})$ denotes the DNN with parameter $\bm{w}$, $G(\cdot;\bm{\theta})$ is the poisoned image generator with parameter $\bm{\theta}$, and $\mathcal{D}=\{(\bm{x}_i, y_i)\}_{i=1}^N$ is a given dataset with $K$ different classes, we have
\begin{equation}\nonumber
\max_{\bm{\theta}} \sum_{i=1}^{N} H\left(f(G(\bm{x}_i;\bm{\theta});\bm{w})\right) < N \cdot \max_{\bm{\theta}} \sum_{j=1}^K \sum_{i=1}^N \mathbb{I}\{y_i=j\} \cdot H\left(\frac{\sum_{i=1}^N f(G(\bm{x}_i;\bm{\theta});\bm{w}) \cdot \mathbb{I}\{y_i=j\}}{\sum_{i=1}^N \mathbb{I}\{y_i=j\}}\right).   
\end{equation}
\end{theorem}

Theorem \ref{thm1} implies that \emph{we can optimize the averaged sample-wise dispersibility $D_s$ and the class-wise dispersibility $D_c$ simultaneously by only maximizing $D_s$}. It motivates us to generate the modified subset $\mathcal{D}_m$ in our UBW-C (via optimizing generator $G$) as follows:

\begin{align}
    &\max_{\bm{\theta}} \sum_{(\bm{x},y) \in \mathcal{D}_s} \left[\mathcal{L}(f(G(\bm{x};\bm{\theta});\bm{w}^{*}), y) + \lambda \cdot H\left(f(G(\bm{x};\bm{\theta});\bm{w}^{*})\right)\right], \\
    & s.t. \ \bm{w}^{*} = \arg \min_{\bm{w}} \sum_{(\bm{x},y)\in \mathcal{D}_p} \mathcal{L}(f(\bm{x};\bm{w}),y), 
\end{align}
where $\lambda$ is a non-negative trade-off hyper-parameter.

In general, the aforementioned process is a standard bi-level optimization, which can be effectively and efficiently solved by alternatively optimizing the lower-level and upper-level sub-problems \cite{liu2021investigating}. In particular, the optimization is conducted via stochastic gradient descent (SGD) with mini-batches \cite{ruder2016overview}, where estimating the class-wise dispersibility is difficult (especially when there are many classes). In contrast, \emph{the estimation of sample-wise dispersibility $D_s$ is still simple and accurate even within a mini-batch}. It is another benefit of only using the averaged sample-wise dispersibility for optimization in our UBW-C. Please refer to the appendix for more our optimization details.

\section{Towards Harmless Dataset Ownership Verification via UBW}

\subsection{Problem Formulation}
Given a suspicious model, the defenders intend to verify whether it is trained on the (protected) dataset. Same as the previous work \cite{li2020open,li2022black}, we assume that the dataset defenders can only query the suspicious model to obtain predicted probability vectors of input samples, whereas having no information about the training process and model parameters.

\subsection{The Proposed Method}

Since defenders can only modify the released dataset and query the suspicious model, the only way to tackle the aforementioned problem is to watermark the (unprotected) benign dataset so that models trained on it will have specific distinctive prediction behaviors. The dataset owners can release the watermarked dataset instead of the original one for copyright protection.

As described in Section \ref{sec:UBW}, the DNNs watermarked by our UBW behave normally on benign samples while having dispersible predictions on poisoned samples. As such, it can be used to design harmless and stealthy dataset ownership verification. 
In general, given a suspicious model, the defenders can verify whether it was trained on the protected dataset by examining whether the model contains specific untargeted backdoor. \emph{The model is regarded as trained on the protected dataset if it contains that backdoor}. To verify it, we design a hypothesis-test-based method, as follows:

\begin{proposition}
Suppose $f(\bm{x})$ is the posterior probability of $\bm{x}$ predicted by the suspicious model. Let variable $\bm{X}$ denotes the benign sample and variable $\bm{X}'$ is its poisoned version ($i.e.$, $\bm{X}'=G(\bm{X})$), while variable $P_b=f(\bm{X})_{Y}$ and $P_p=f(\bm{X}')_{Y}$ indicate the predicted probability on the ground-truth label $Y$ of $\bm{X}$ and $\bm{X}'$, respectively. Given the null hypothesis $H_0: P_b = P_p + \tau$ ($H_1: P_b  > P_p + \tau$) where the hyper-parameter $\tau \in [0,1]$, we claim that the suspicious model is trained on the protected dataset (with $\tau$-certainty) if and only if $H_0$ is rejected. 
\end{proposition}

In practice, we randomly sample $m$ different benign samples to conduct the pairwise T-test \cite{hogg2005introduction} and calculate its p-value. The null hypothesis $H_0$ is rejected if the p-value is smaller than the significance level $\alpha$. In particular, \emph{we only select samples that can be correctly classified by the suspicious model} to reduce the side-effects of model accuracy. Otherwise, due to the untargeted nature of our UBW, our verification may misjudge when there is dataset stealing, if the benign accuracy of the suspicious model is relatively low. Besides, we also calculate the \emph{confidence score} $\Delta P = P_b - P_p$ to represent the verification confidence. \emph{The larger the $\Delta P$, the more confident the verification}.

\section{Experiments}
\label{sec:exps}

\subsection{Experimental Settings}

\subsec{Datasets and Models. }
In this paper, we conduct experiments on two classical benchmark datasets, including CIFAR-10 \cite{krizhevsky2009learning} and (a subset of) ImageNet \cite{deng2009imagenet}, with ResNet-18 \cite{he2016deep}. Specifically, we randomly select a subset containing $50$ classes with $25,000$ images from the original ImageNet for training (500 images per class) and $2,500$ images for testing (50 images per class). For simplicity, all images are resized to $3 \times 64 \times 64$, following the settings used in Tiny-ImageNet \cite{chrabaszcz2017downsampled}.

\subsec{Baseline Selection. }
We compare our UBW with representative existing poison-only backdoor attacks. Specifically, for attacks with poisoned labels, we adopt BadNets \cite{gu2019badnets}, blended attack (dubbed as `Blended') \cite{chen2017targeted}, and WaNet \cite{nguyen2021wanet} as the baseline methods. They are the representative of visible attacks, patch-based invisible attacks, and non-patch-based invisible attacks, respectively. We use the label-consistent attack (dubbed as `Label-Consistent') \cite{turner2019label} and Sleeper Agent \cite{souri2022sleeper} as the representative of attacks with clean labels. Besides, we also include the models trained on the benign dataset (dubbed as `No Attack') as another baseline for reference.

\subsection{The Performance of Dataset Watermarking}

\subsec{Settings. }
We set the poisoning rate $\gamma = 0.1$ for all watermarks on both datasets. In particular, since the label-consistent attack can only modify samples from the target class, its poisoning rate is set to its maximum ($i.e.$, 0.02) on the ImageNet dataset. The target label $y_t$ is set to 1 for all targeted watermarks. Besides, following the classical settings in existing papers, we adopt a white-black square as the trigger pattern for BadNets, blended attack, label-consistent attack, and UBW-P on both datasets. The trigger patterns adopted for Sleeper Agent and UBW-C are sample-specific. We set $\lambda=2$ for UBW-C on both datasets. The example of poisoned samples generated by different methods is shown in Figure \ref{fig:poisoned_samples}. More detailed settings are described in the appendix.

\subsec{Evaluation Metrics. }
We use the benign accuracy (BA), the attack success rate (ASR), and the averaged prediction dispersibility ($D_p$) to evaluate the watermark performance. In particular, we introduce two types of ASR, including the attack success rate on all testing samples (ASR-A) and the attack success rate on correctly classified testing samples (ASR-C). In general, \emph{the larger the BA, ASR, and $D_p$, the better the watermark}. Please refer to the appendix for more details.

\begin{figure}
\vspace{-2.5em}
\centering
\subfigure[CIFAR-10]{
\centering
\includegraphics[width=0.83\textwidth]{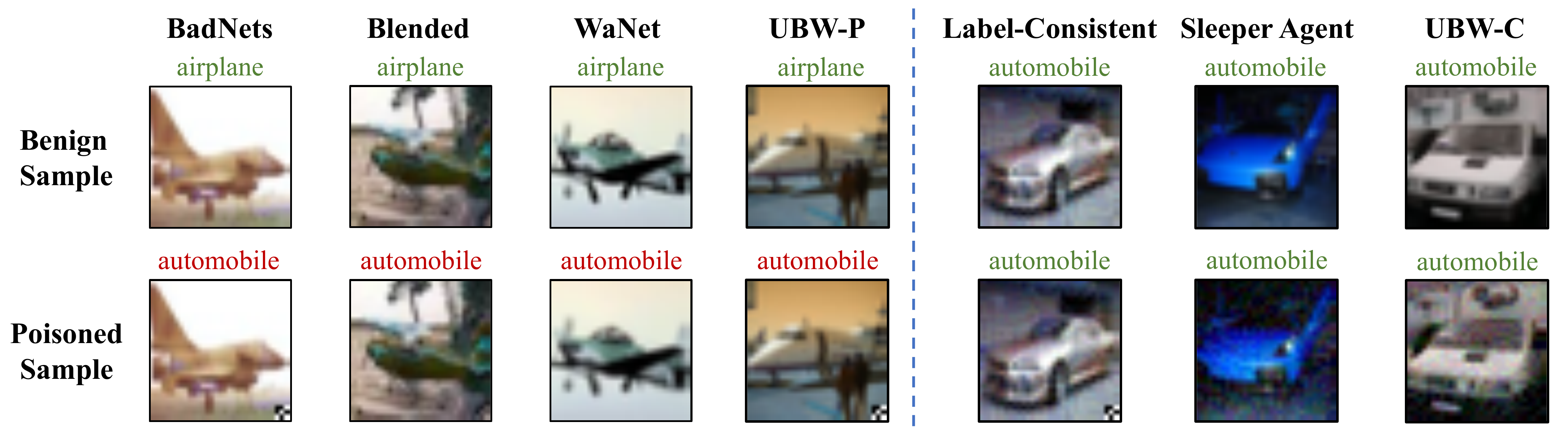}
}\vspace{-0.7em}
\subfigure[ImageNet]{
\centering
\includegraphics[width=0.83\textwidth]{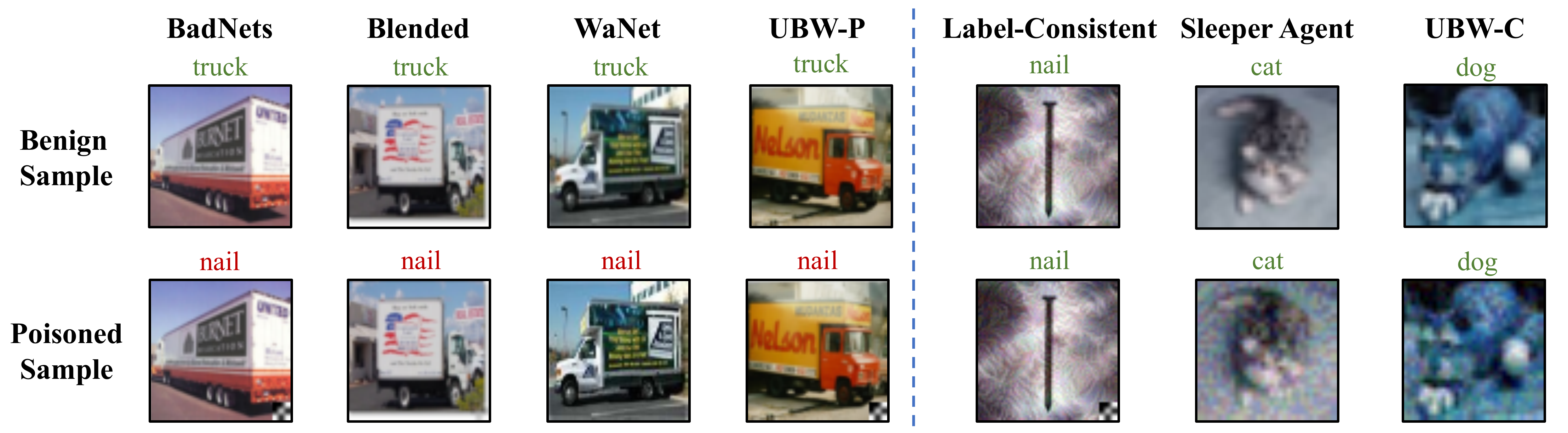}
}
\vspace{-0.7em}
\caption{The example of samples involved in different backdoor watermarks. In the BadNets, blended attack, WaNet, and UBW-P, the labels of poisoned samples are inconsistent with their ground-truth ones. In the label-consistent attack, Sleeper Agent, and UBW-C, the labels of poisoned samples are the same as their ground-truth ones. In particular, the label-consistent attack can only poison samples in the target class, while other methods can modify all samples. }
\label{fig:poisoned_samples}
\vspace{-0.7em}
\end{figure}

\begin{table}[t]
\centering
\caption{The watermark performance on the CIFAR-10 dataset.}
\vspace{-0.5em}
\scalebox{0.92}{
\begin{tabular}{cc|c|ccc|c}
\toprule
\multicolumn{1}{c|}{Label Type$\downarrow$}                    & Target Type$\downarrow$               & Method$\downarrow$, Metric$\rightarrow$   & BA (\%) & ASR-A (\%) & ASR-C (\%) & $D_p$ \\ \hline
\multicolumn{2}{c|}{N/A}                                                       & No Attack        &  92.53  & N/A    & N/A      &  N/A     \\ \hline
\multicolumn{1}{c|}{\multirow{4}{*}{Poisoned-Label}} & \multirow{3}{*}{Targeted} & BadNets          &  91.52  &   100  &   100    &  0.0000    \\
\multicolumn{1}{c|}{}                              &                           & Blended      & 91.61   &   100  &   100    &   0.0000    \\
\multicolumn{1}{c|}{}                              &                           & WaNet            & 90.48   &   95.50  &   95.33    &  0.1979     \\ \cline{2-7} 
\multicolumn{1}{c|}{}              & Untargeted                & UBW-P (Ours)    & 90.59   &  92.30   &   92.51    &  2.2548    \\ \hline \hline
\multicolumn{1}{c|}{\multirow{3}{*}{Clean-Label}}  & \multirow{2}{*}{Targeted} & Label-Consistent &  82.94  &  96.00  &   95.80    &   0.9280    \\
\multicolumn{1}{c|}{}                              &                           & Sleeper Agent     &  86.06  &  70.60   &   54.46    &   1.0082    \\ \cline{2-7} 
\multicolumn{1}{c|}{}                              & Untargeted                & UBW-C (Ours)   &  86.99 &  89.80  &  87.56   &   1.2641   \\ \bottomrule
\end{tabular}
}
\label{tab:attack_cifar}
\vspace{-1.1em}
\end{table}

\subsec{Results. }
As shown in Table \ref{tab:attack_cifar}-\ref{tab:attack_imagenet}, \emph{the performance of our UBW is on par with that of baseline targeted backdoor watermarks} under both poisoned-label and clean-label settings. Especially under the clean-label setting, our \emph{UBW-C is significantly better than other watermarks with clean labels}. For example, the ASR-C increases of our method compared with label-consistent attack and Sleeper Agent are both over 55\% on ImageNet. These results verify that our \emph{UBW can implant distinctive behaviors in attacked DNNs}. In particular, our UBW has significantly higher averaged prediction dispersibility $D_p$, especially under the poisoned-label setting. For example, the $D_p$ of UBW-P is more than 10 times larger than that of all baseline attacks with poisoned labels on the CIFAR-10 dataset. These results verify that the \emph{UBW can not manipulate malicious predictions deterministically and therefore is harmless}. Moreover, we notice that the $D_p$ of label-consistent attack and Sleeper Agent is similar to that of our UBW-C to some extent. It is mostly because targeted attacks with clean labels are significantly more difficult in making all poisoned samples to the same (target) class. 

\begin{table}[!t]
\centering
\vspace{-2em}
\caption{The watermark performance on the ImageNet dataset.}
\vspace{-0.5em}
\scalebox{0.92}{
\begin{tabular}{cc|c|ccc|c}
\toprule
\multicolumn{1}{c|}{Label Type$\downarrow$}                    & Target Type$\downarrow$               & Method$\downarrow$, Metric$\rightarrow$   & BA (\%) & ASR-A (\%) & ASR-C (\%) & $D_p$ \\ \hline
\multicolumn{2}{c|}{N/A}                                                       & No Attack        &  67.30  &  N/A   &  N/A     &   N/A    \\ \hline
\multicolumn{1}{c|}{\multirow{4}{*}{Poisoned-Label}} & \multirow{3}{*}{Targeted} & BadNets          &  65.64  &  100  &   100    &    0.0000   \\
\multicolumn{1}{c|}{}                              &                           & Blended      &  65.28  &  88.00  &   85.37    &   0.3669    \\
\multicolumn{1}{c|}{}                              &                           & WaNet    &  62.56  &  78.00  &   73.17    &   0.7124    \\ \cline{2-7} 
\multicolumn{1}{c|}{}                              & Untargeted                & UBW-P (Ours)    &  62.60  &  82.00   &   82.61    &  2.7156     \\ \hline \hline
\multicolumn{1}{c|}{\multirow{3}{*}{Clean-Label}}  & \multirow{2}{*}{Targeted} & Label-Consistent&  62.36  &  30.00  & 2.78      &   1.2187    \\
\multicolumn{1}{c|}{}                              &                           & Sleeper Agent     &  56.92  &   6.00  &   2.31    &  1.0943     \\ \cline{2-7} 
\multicolumn{1}{c|}{}                              & Untargeted                & UBW-C (Ours)  &  59.64  &  74.00   &   60.00    &  2.4010     \\ \bottomrule
\end{tabular}
}
\vspace{-0.6em}
\label{tab:attack_imagenet}
\end{table}

\subsection{The Performance of UBW-based Dataset Ownership Verification}

\subsec{Settings. }
We evaluate our verification method in three representative scenarios, including \textbf{1)} independent trigger (dubbed as `Independent-T'), \textbf{2)} independent model (dubbed as `Independent-M'), and \textbf{3)} unauthorized dataset usage (dubbed as `Malicious'). In the first scenario, we query the attacked suspicious model using the trigger that is different from the one used for model training; In the second scenario, we examine the benign suspicious model using the trigger pattern; We adopt the trigger used in the training process of the watermarked suspicious model in the last scenario. We set $\tau = 0.25$ for the hypothesis-test in all cases. More detailed settings are in the appendix.

\subsec{Evaluation Metrics. }
We adopt the $\Delta P \in [-1,1]$ and the p-value $\in [0,1]$ for the evaluation. For the two independent scenarios, the smaller the $\Delta P$ and the larger the p-value, the better the verification; For the malicious one, the larger the $\Delta P$ and the smaller the p-value, the better the verification.

\subsec{Results. }
As shown in Table \ref{tab:verification_p}-\ref{tab:verification_c}, our dataset ownership verification is effective in all cases, no matter under UBW-P or UBW-C. Specifically, our method can accurately identify unauthorized dataset usage ($i.e.$, `Malicious') with high confidence ($i.e.$, $\Delta P \gg 0$ and p-value $\ll 0.01$) while does not misjudge ($i.e.$, $\Delta P$ is nearly 0 and p-value $\gg 0.05$) when there is no stealing ($i.e.$, `Independent-T' and `Independent-M'). For example, the p-values of verifying independent cases are all nearly 1 on both datasets. We notice that the verification performance under UBW-C is relatively poorer than that under UBW-P, although its performance is already capable enough for verification. However, the UBW-C is more stealthy, since the labels of poisoned samples are consistent with their ground-truth label and the trigger patterns are invisible. Users can adopt different UBWs based on their needs.

\begin{table}[!t]
\centering
\caption{The effectiveness of dataset ownership verification via UBW-P.}
\vspace{-0.5em}
\scalebox{0.9}{
\begin{tabular}{c|ccc|ccc}
\toprule
& \multicolumn{3}{c|}{CIFAR-10}        & \multicolumn{3}{c}{ImageNet}         \\ \hline
 & Independent-T & Independent-M & Malicious & Independent-T & Independent-M & Malicious \\ \hline
$\Delta P$ &       -0.0269       &       0.0024        &      0.7568     &     0.1281          &     0.0241          &     0.8000      \\
p-value      &   1.0000            &   1.0000        &     $10^{-36}$      &       0.9666        &     1.0000          &     $10^{-10}$      \\ \bottomrule
\end{tabular}
}
\vspace{-0.8em}
\label{tab:verification_p}
\end{table}

\begin{table}[!t]
\centering
\caption{The effectiveness of dataset ownership verification via UBW-C.}
\vspace{-0.5em}
\scalebox{0.9}{
\begin{tabular}{c|ccc|ccc}
\toprule
    &\multicolumn{3}{c|}{CIFAR-10}             & \multicolumn{3}{c}{ImageNet}              \\ \hline
   & Independent-T & Independent-M & Malicious & Independent-T & Independent-M & Malicious \\ \hline
$\Delta P$ &      0.1874         &      0.0171         &     0.6115     &     0.0588         &        0.1361       &  0.4836         \\
p-value        &    0.9688          &      1.0000         &    $10^{-14}$      &   0.9999    &    0.9556           &     0.0032      \\ \bottomrule
\end{tabular}
}
\vspace{-0.8em}
\label{tab:verification_c}
\end{table}

\subsection{Discussion}
\subsubsection{The Ablation Study}
In this section, we explore the effects of key hyper-parameters involved in our UBW. The detailed settings and the effects of hyper-parameters involved in ownership verification are in the appendix.

\begin{figure}[ht]
\centering
\vspace{-2em}
\begin{minipage}[t]{0.48\linewidth}
\centering
\includegraphics[width=\textwidth]{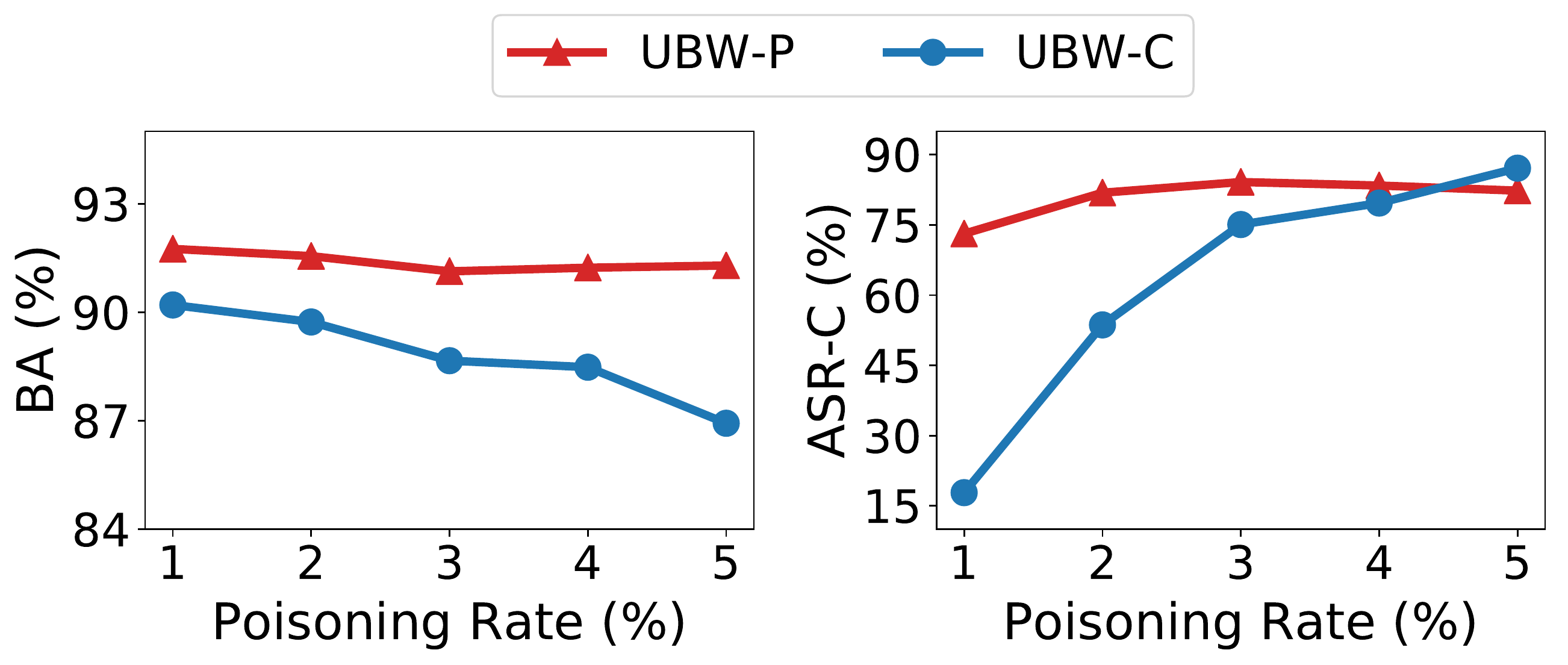}
\vspace{-1.3em}
\caption{The effects of poisoning rate $\gamma$.}
\label{fig:effects_gamma}
\end{minipage}\hspace{0.3em}
\begin{minipage}[t]{0.48\linewidth}
\centering
\includegraphics[width=\textwidth]{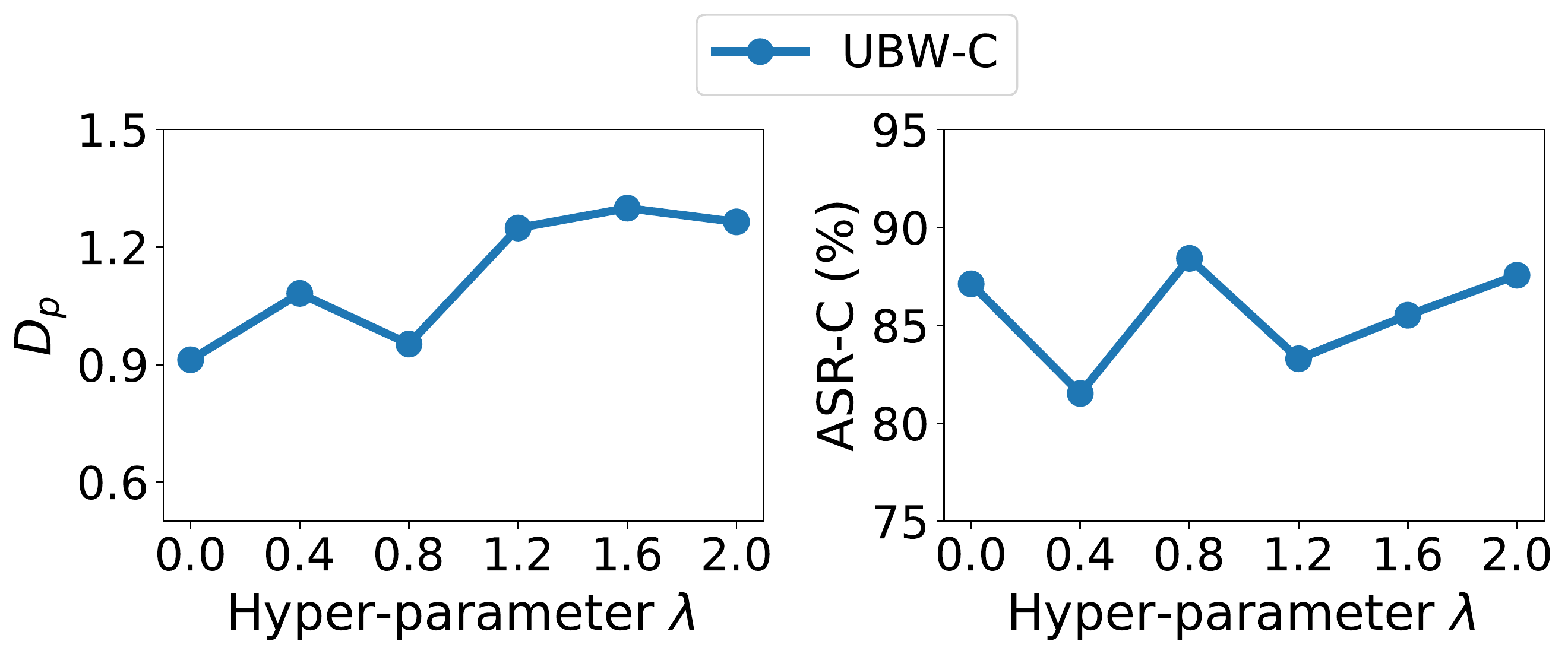}
\vspace{-1.3em}
\caption{The effects of hyper-parameter $\lambda$.}
\label{fig:effects_lambda}
\end{minipage}
\end{figure}

\begin{figure}[ht]
\centering
\vspace{-0.7em}
\begin{minipage}[t]{0.48\linewidth}
\centering
\includegraphics[width=\textwidth]{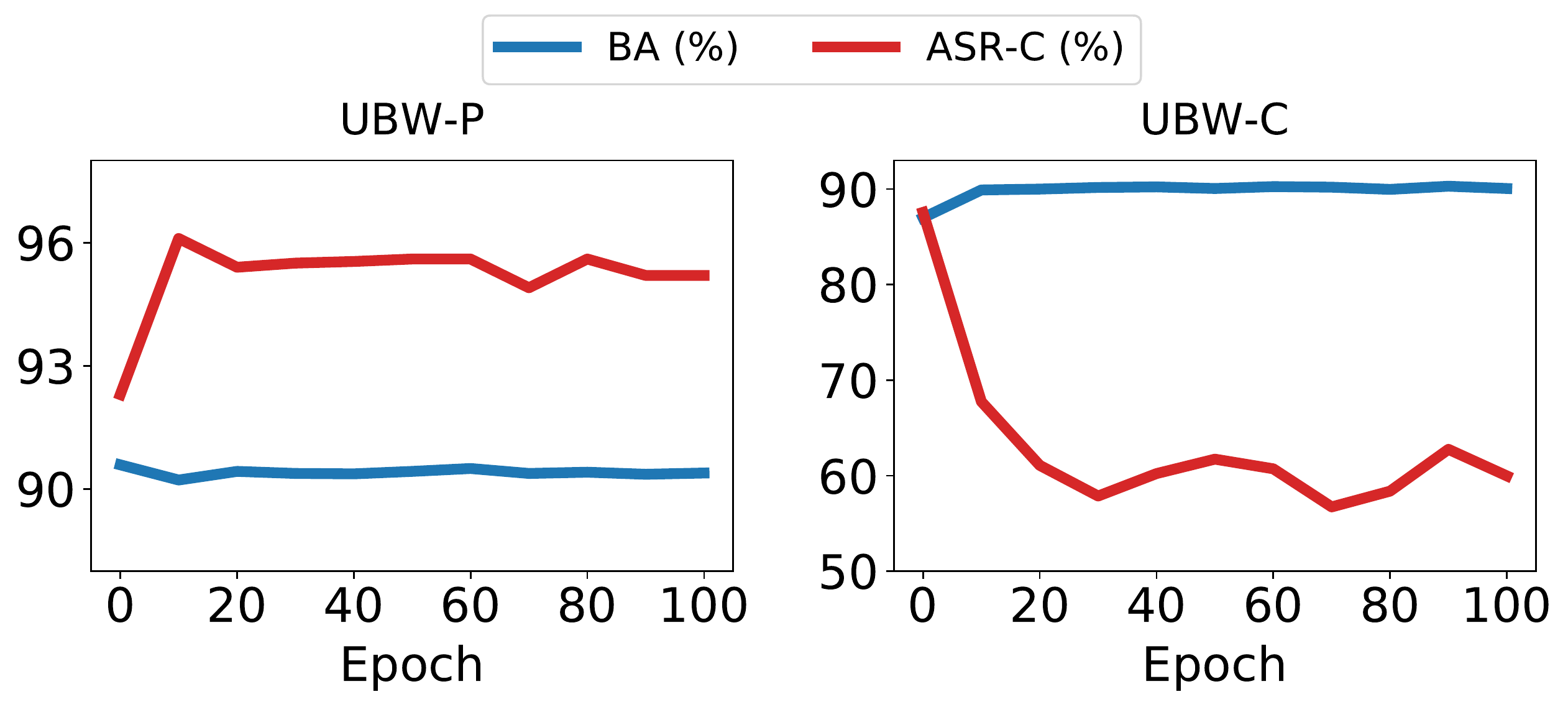}
\vspace{-1.3em}
\caption{The resistance to fine-tuning.}
\label{fig:resistance_tuning}
\end{minipage}\hspace{0.3em}
\begin{minipage}[t]{0.48\linewidth}
\centering
\includegraphics[width=\textwidth]{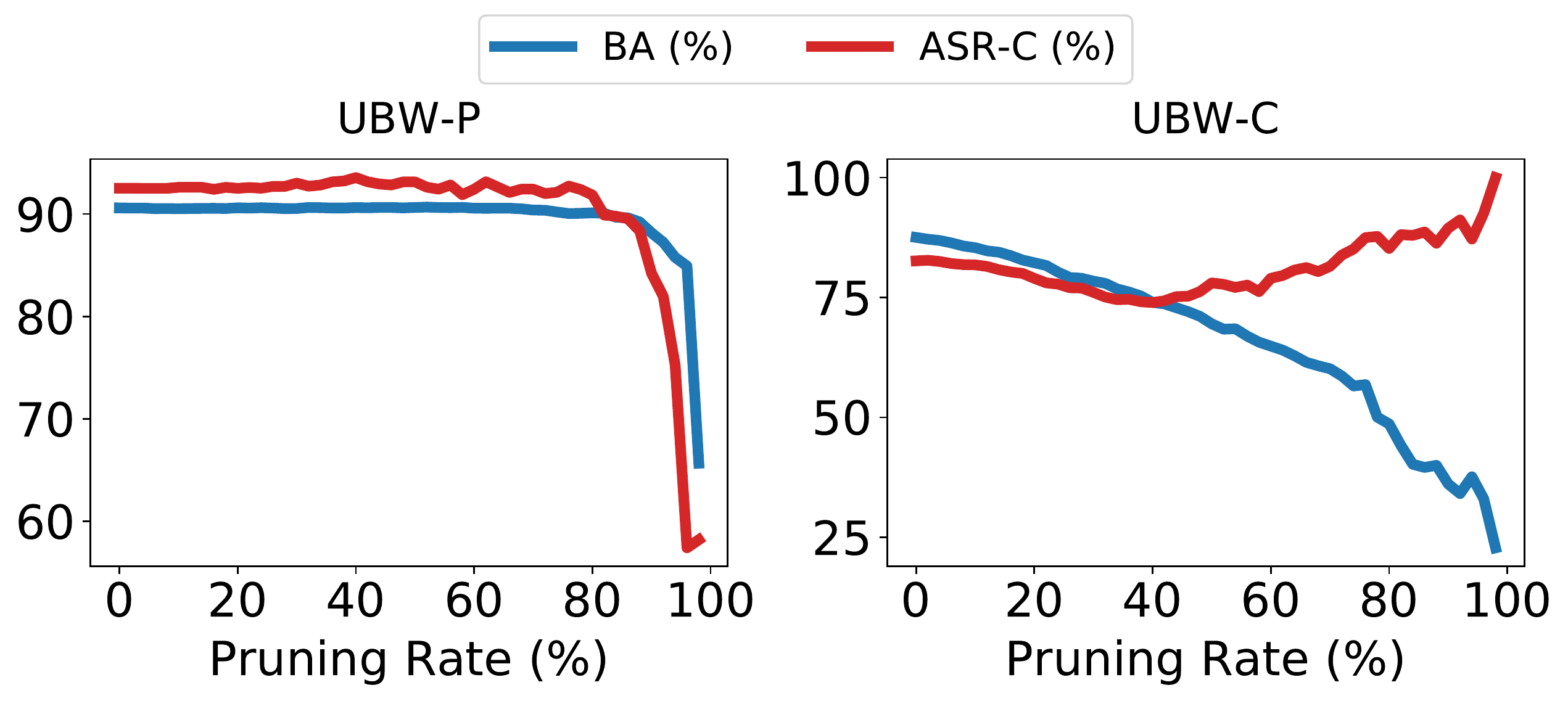}
\vspace{-1.3em}
\caption{The resistance to model pruning.}
\label{fig:resistance_pruning}
\end{minipage}
\vspace{-1em}
\end{figure}

\subsec{Effects of Poisoning Rate $\gamma$. }
As shown in Figure \ref{fig:effects_gamma}, the attack success rate (ASR) increases with the increase of the poisoning rate $\gamma$. Both UBW-P and UBW-C reach promising ASR even when $\gamma$ is small ($e.g.$, 0.03). Besides, the benign accuracy decreases with the increase of $\gamma$. Users should assign $\gamma$ based on their specific requirements in practice. 

\subsec{Effects of Trade-off Hyper-parameter $\lambda$. }
As shown in Figure \ref{fig:effects_lambda}, the averaged prediction dispersibility $D_p$ increases with the increase of $\lambda$. This phenomenon indicates that the averaged sample-wise dispersibility $D_s$ used in our UBW-C is a good approximation of $D_p$. In contrast, increasing $\lambda$ has minor effects in ASR, which is probably because the untargeted attack scheme is more stable.

\subsubsection{Resistance to Backdoor Defenses}
\label{sec:resistance_defenses}
In this section, we discuss whether our UBW is resistant to existing backdoor defenses so that it can still provide promising dataset protection even under adaptive opposite methods. In particular, the trigger patterns used by our UBW-C are \emph{sample-specific}, where different poisoned images contain different triggers (as shown in Figure \ref{fig:poisoned_samples}). Recently, ISSBA \cite{li2021invisible} revealed that most of the existing defenses ($e.g.$, Neural Cleanse \cite{wang2019neural}, SentiNet \cite{chou2020sentinet}, and STRIP \cite{gao2021design}) have a latent assumption that the trigger patterns are \emph{sample-agnostic}. Accordingly, our \emph{UBW-C can naturally bypass them}, since it breaks their fundamental assumption. Here we explore the resistance of our UBW to fine-tuning \cite{liu2017neural,liu2018fine} and model pruning \cite{liu2018fine,wu2021adversarial}, which are the representative defenses whose effects did not rely on this assumption. The detailed settings and resistance to other defenses are in the appendix.

As shown in Figure \ref{fig:resistance_tuning}, our \emph{UBW is resistant to fine-tuning}. Specifically, the attack success rates are still larger than $55\%$ for both UBW-P and UBW-C after the fine-tuning process is finished. Besides, our \emph{UBW is also resistant to model pruning} (as shown in Figure \ref{fig:resistance_pruning}). The ASRs of both UBW-P and UBW-C are larger than $50\%$ even under high pruning rates, where the benign accuracies are already low. An interesting phenomenon is that as the pruning rate increases, the ASR of UBW-C even increases for a period. We speculate that it is probably because our UBW-C is untargeted and sample-specific, and therefore it can reach better attack effects when the model's benign functions are significantly depressed. We will further discuss its mechanism in our future work.

\section{Societal Impacts}
\label{sec:impacts}
This paper is the first attempt toward untargeted backdoor attacks and their positive applications. In general, our main focus is how to design and use untargeted backdoor attacks as harmless and stealthy watermarks for dataset protection, which has positive societal impacts. We notice that our untargeted backdoor watermark (UBW) is resistant to existing backdoor defenses and could be maliciously used by the backdoor adversaries. However, compared with existing targeted backdoor attacks, our UBW is untargeted and therefore has minor threats. Moreover, although an effective defense is yet to be developed, people can still mitigate or even avoid the threats by only using trusted training resources. 

\section{Conclusion}
In this paper, we revisited how to protect the copyrights of (open-sourced) datasets. We revealed that existing dataset ownership verification could introduce new serious risks, due to the targeted nature of existing poison-only backdoor attacks used for dataset watermarking. Based on this understanding, we explored the untargeted backdoor watermark (UBW) paradigm under both poisoned-label and clean-label settings, whose abnormal model behaviors were not deterministic. We also studied how to exploit our UBW for harmless and stealthy dataset ownership verification. Experiments on benchmark datasets validated the effectiveness of our method and its resistance to backdoor defenses.

\section*{Acknowledgments}
This work is supported in part by the National Natural Science Foundation of China under Grant 62171248, the PCNL Key Project (PCL2021A07), the Tencent Rhino-Bird Research Program, and the C3 AI and Amazon research awards. We also sincerely thank Linghui Zhu from Tsinghua University for her assistance in the experiments of resistance to saliency-based backdoor defenses.

\bibliographystyle{unsrt}
\bibliography{reference}

\newpage

\section*{Checklist}

\begin{enumerate}

\item For all authors...
\begin{enumerate}
  \item Do the main claims made in the abstract and introduction accurately reflect the paper's contributions and scope?
    \answerYes{\textbf{The main claims met the contribution part in the introduction.} } 
  \item Did you describe the limitations of your work?
    \answerYes{\textbf{We discussed them in Section \ref{sec:impacts}.}} 
  \item Did you discuss any potential negative societal impacts of your work?
    \answerYes{\textbf{We discussed them in Section \ref{sec:impacts}.}}
  \item Have you read the ethics review guidelines and ensured that your paper conforms to them?
    \answerYes{\textbf{We read and ensured that our paper conforms to them.}} 
\end{enumerate}

\item If you are including theoretical results...
\begin{enumerate}
  \item Did you state the full set of assumptions of all theoretical results?
        \answerYes{\textbf{We included them in Section \ref{sec:UBW-C} and the appendix. }} 
  \item Did you include complete proofs of all theoretical results?
        \answerYes{\textbf{We included them in the appendix.}} 
\end{enumerate}

\item If you ran experiments...
\begin{enumerate}
  \item Did you include the code, data, and instructions needed to reproduce the main experimental results (either in the supplemental material or as a URL)?
    \answerYes{\textbf{We included them in the appendix. }}
  \item Did you specify all the training details (e.g., data splits, hyperparameters, how they were chosen)?
    \answerYes{\textbf{We included them in the appendix. }}
	\item Did you report error bars (e.g., with respect to the random seed after running experiments multiple times)?
    \answerNo{\textbf{Due to the limitation of time and computational resources, we did not report it in the submission.}}
	\item Did you include the total amount of compute and the type of resources used (e.g., type of GPUs, internal cluster, or cloud provider)?
    \answerYes{\textbf{We included them in the appendix. }}
\end{enumerate}

\item If you are using existing assets (e.g., code, data, models) or curating/releasing new assets...
\begin{enumerate}
  \item If your work uses existing assets, did you cite the creators?
    \answerYes{\textbf{We cited their papers or provided their links. }}
  \item Did you mention the license of the assets?
    \answerYes{\textbf{We included them in the appendix and we also cited their papers or provided their links. }}
  \item Did you include any new assets either in the supplemental material or as a URL?
    \answerYes{\textbf{We provided the codes and trained models in the appendix. }}
  \item Did you discuss whether and how consent was obtained from people whose data you're using/curating?
    \answerYes{\textbf{We discussed it in the appendix. }}
  \item Did you discuss whether the data you are using/curating contains personally identifiable information or offensive content?
    \answerNA{\bf We did not use the data which contains personally identifiable information or offensive content.}
\end{enumerate}

\item If you used crowdsourcing or conducted research with human subjects...
\begin{enumerate}
  \item Did you include the full text of instructions given to participants and screenshots, if applicable?
    \answerNA{\textbf{We did not use crowdsourcing or conduct research with human subjects. }}
  \item Did you describe any potential participant risks, with links to Institutional Review Board (IRB) approvals, if applicable?
    \answerNA{\textbf{We did not use crowdsourcing or conduct research with human subjects. }}
  \item Did you include the estimated hourly wage paid to participants and the total amount spent on participant compensation?
    \answerNA{\textbf{We did not use crowdsourcing or conduct research with human subjects. }}
\end{enumerate}

\end{enumerate}

\newpage

\appendix
\setcounter{equation}{0}
\setcounter{lemma}{0}
\setcounter{theorem}{0}

\section{The Omitted Proofs}

\begin{lemma}
\label{appendix:lemma1}
The averaged class-wise dispersibility is always greater than the averaged sample-wise dispersibility divided by $N$, $i.e.$, 
$D_c > \frac{1}{N} \cdot D_s$.
\end{lemma}

\begin{proof}
Since entropy is concave \cite{kullback1997information}, according to Jensen's inequality, we have: 

\begin{equation}\label{eq:jensen}
\begin{aligned}
    H\left(\frac{\sum_{k=1}^N f(\bm{x}_k) \cdot \mathbb{I}\{y_k=j\}}{\sum_{k=1}^N \mathbb{I}\{y_k=j\}}\right) & \geq 
    \sum_{k=1}^N \frac{\mathbb{I}\{y_k=j\}}{\sum_{k=1}^N \mathbb{I}\{y_k=j\}} H\left(f(\bm{x}_k)\right) \\ 
    & > \frac{1}{N} \cdot \sum_{k=1}^N \mathbb{I}\{y_k=j\} \cdot H\left(f(\bm{x}_k)\right).
\end{aligned}
\end{equation}

Since each sample $\bm{x}$ has and only has one label $y \in \{1, \cdots, K \}$, we have:
\begin{equation}\label{eq:entropy}
    H\left(f(\bm{x}_i)\right) = \sum_{j=1}^K \mathbb{I}\{y_i=j\} \cdot H\left(f(\bm{x}_i)\right), \forall i \in \{1, \cdots, N\}.
\end{equation}

According to equation (\ref{eq:jensen}) and the definition of $D_c$, we have

\begin{equation}\label{eq:D_c}
    D_c > \frac{1}{N^2} \sum_{j=1}^K \sum_{i=1}^N \sum_{k=1}^N \mathbb{I}\{y_i=j\} \cdot \mathbb{I}\{y_k=j\} \cdot H\left(f(\bm{x}_k)\right).
\end{equation}

Due to the property of the indicator function $\mathbb{I}\{\cdot\}$, we have

\begin{equation}\label{eq:indicator}
\sum_{i=1}^N \sum_{k=1}^N \mathbb{I}\{y_i=j\} \cdot \mathbb{I}\{y_k=j\} \cdot H\left(f(\bm{x}_k)\right) \geq \sum_{i=1}^N \mathbb{I}\{y_i=j\} \cdot H\left(f(\bm{x}_i)\right).    
\end{equation}

According to equation (\ref{eq:entropy}) to equation (\ref{eq:indicator}), we have

\begin{equation}
   D_c > \frac{1}{N^2} \sum_{j=1}^K \sum_{i=1}^N \mathbb{I}\{y_i=j\} \cdot H\left(f(\bm{x}_i)\right) = \frac{1}{N^2} \sum_{i=1}^N H\left(f(\bm{x}_i)\right) = \frac{1}{N} D_s.
\end{equation}

\end{proof}

\begin{theorem}
\label{appendix:thm1}
Let $f(\cdot;\bm{w})$ denotes the DNN with parameter $\bm{w}$, $G(\cdot;\bm{\theta})$ is the poisoned image generator with parameter $\bm{\theta}$, and $\mathcal{D}=\{(\bm{x}_i, y_i)\}_{i=1}^N$ is a given dataset with $K$ different classes, we have
\begin{equation}\nonumber
\max_{\bm{\theta}} \sum_{i=1}^{N} H\left(f(G(\bm{x}_i;\bm{\theta});\bm{w})\right) < N \cdot \max_{\bm{\theta}} \sum_{j=1}^K \sum_{i=1}^N \mathbb{I}\{y_i=j\} \cdot H\left(\frac{\sum_{i=1}^N f(G(\bm{x}_i;\bm{\theta});\bm{w}) \cdot \mathbb{I}\{y_i=j\}}{\sum_{i=1}^N \mathbb{I}\{y_i=j\}}\right).   
\end{equation}
\end{theorem}

\begin{proof}
The proof is straightforward given Lemma \ref{lemma1}, based on replacing $f(\bm{x}_i)$ with $f(G(\bm{x}_i;\bm{\theta});\bm{w})$ while multiplying both sides by $N$ and maximizing both sides.  
\end{proof}

\section{The Optimization Process of our UBW-C}
\label{sec:opt_details}
Recall that the optimization objective of our UBW-C is as follows:

\begin{align}
    &\max_{\bm{\theta}} \sum_{(\bm{x},y) \in \mathcal{D}_s} \left[\mathcal{L}(f(G(\bm{x};\bm{\theta});\bm{w}^{*}), y) + \lambda \cdot H\left(f(G(\bm{x};\bm{\theta});\bm{w}^{*})\right)\right], \\
    & s.t. \ \bm{w}^{*} = \arg \min_{\bm{w}} \sum_{(\bm{x},y)\in \mathcal{D}_p} \mathcal{L}(f(\bm{x};\bm{w}),y), 
\end{align}
where $\lambda$ is a non-negative trade-off hyper-parameter.

In general, the aforementioned process is a standard bi-level optimization, which can be effectively and efficiently solved by alternatively optimizing the upper-level and lower-level sub-problems \cite{liu2021investigating}. To solve the aforementioned problem, the form of $G$ is one of the key factors. Inspired by the hidden trigger backdoor attack \cite{saha2020hidden} and the Sleeper Agent \cite{souri2022sleeper}, we also adopt different generators during the training and inference process to enhance attack effectiveness and stealthiness, as follows:

Let $G_t$ and $G_i$ denote the generator used in the training and inference process, respectively. We intend to generate \emph{sample-specific} small additive perturbations for selected training images based on $G_t$ so that their gradient ensemble has a similar direction to the gradient ensemble of poisoned `testing' images generated by $G_i$. Specifically, we set $G_t(\bm{x}) = \bm{x} + \bm{\theta}(\bm{x})$, where $||\bm{\theta}(\bm{x})||_{\infty} \leq \epsilon$ and $\epsilon$ is the perturbation budget; We set $G_i(\bm{x}) = (\bm{1}-\bm{\alpha}) \otimes \bm{x} + \bm{\alpha} \otimes \bm{t}$, where $\bm{\alpha} \in \{0,1\}^{C \times W \times H}$ denotes the given mask and $\bm{t} \in \mathcal{X}$ is the given trigger pattern. In general, the trigger patterns used for training is invisible for stealthiness while those used for inference is visible for effectiveness. The detailed lower-level and upper-level sub-problems are as follows:

\subsec{Upper-level Sub-problem. }
Given the current model parameters $\bm{w}$, we optimize the trigger patterns $\{\bm{\theta}(\bm{x})|\bm{x} \in \mathcal{D}_s \}$ of selected training samples (for poisoning) based on the gradient matching:

\begin{equation}
    \max_{\{\bm{\theta}(\bm{x})|\bm{x} \in \mathcal{D}_s,\  ||\bm{\theta}(\bm{x})||_{\infty} \leq \epsilon \}} \frac{\nabla_{\bm{w}} \mathcal{L}_t \cdot \nabla_{\bm{w}} \mathcal{L}_i}{||\mathcal{L}_t|| \cdot ||\mathcal{L}_i||},
\end{equation}
where 
\begin{equation}\label{eq:poisoned_test}
    \mathcal{L}_i = \frac{1}{N} \cdot \sum_{(\bm{x},y) \in \mathcal{D}} \left[\mathcal{L}(f(G_i(\bm{x});\bm{w}), y) + \lambda \cdot H\left(f(G_i(\bm{x});\bm{w})\right)\right], 
\end{equation}

\begin{equation}
    \mathcal{L}_t = \frac{1}{M} \cdot \sum_{(\bm{x},y) \in \mathcal{D}_s} \mathcal{L}(f(\bm{x}+\bm{\theta}(\bm{x});\bm{w}), y), 
\end{equation}
$N$ and $M$ denote the number of training samples and the number of selected samples, respectively. The upper-level sub-problem is solved by projected gradient ascend (PGA) \cite{ruder2016overview}.

\subsec{Lower-level Sub-problem. }
Given the current trigger patterns $\{\bm{\theta}(\bm{x})|\bm{x} \in \mathcal{D}_s \}$, we can obtain the poisoned training dataset $\mathcal{D}_p$ and then optimize the model parameters $\bm{w}$ via   

\begin{equation}
    \min_{\bm{w}} \sum_{(\bm{x},y)\in \mathcal{D}_p} \mathcal{L}(f(\bm{x};\bm{w}),y).
\end{equation}
The lower-level sub-problem is solved by stochastic gradient descent (SGD) \cite{ruder2016overview}.

\vspace{0.2em}
Besides, there are three additional optimization details that we need to mention, as follows:

\subsec{1) How to Select Training Samples for Poisoning. }
We select training samples with the largest gradient norms instead of random selection for poisoning since they have more influence. It is allowed in our UBW since the dataset owner can determine which samples should be modified.

\subsec{2) How to Select `Test' Samples for Poisoning. }
Instead of using all training samples to calculate Eq. (\ref{eq:poisoned_test}), we only use those from a specific source class. This approach is used to further enhance UBW effectiveness, since the gradient ensemble of samples from all classes may be too `noisy' to learn for $G_t$. Its benefits are verified in the following Section \ref{sec:effects_all}.

\subsec{3) The Relation between Dispersibility and Attack Success Rate. }
In general, the optimization of dispersibility contradicts to that of the attack success rate to some extent. Specifically, let us consider a classification problem with $K$ different classes. When the averaged sample-wise dispersibility used in optimizing UBW-C reaches its maximum value, the attack success rate is only $\frac{K-1}{K}$, since the predicted probability vectors are all uniform; When the attack success rate reaches 100\%, both averaged prediction dispersibility and sample-wise dispersibility cannot reach their maximum.

In particular, similar to other backdoor attacks based on bi-level optimization ($e.g.$, LIRA \cite{doan2021lira} and Sleeper Agent \cite{souri2022sleeper}), we notice that the watermark performance of our UBW-C is not very stable across different random seeds ($i.e.$, has relatively large standard deviation). We will explore how to stabilize and improve the performance of UBW-C in our future work.

\begin{figure}[!t]
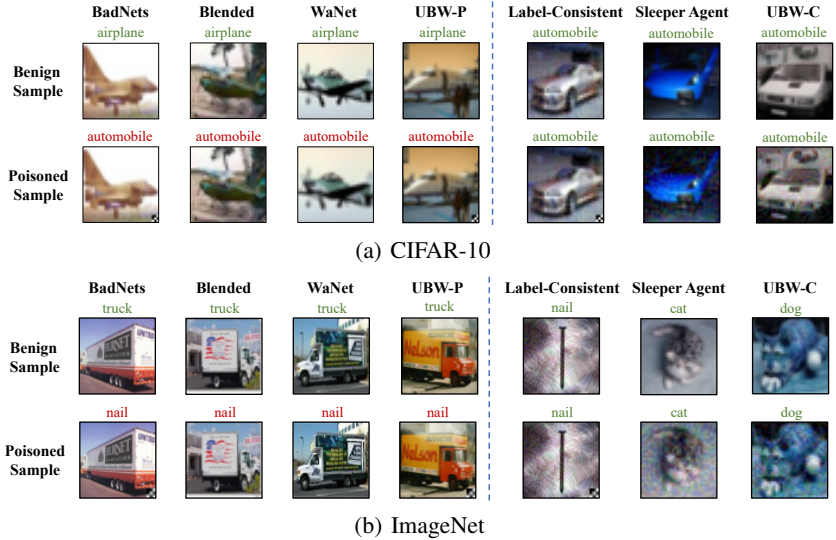

\vspace{-2.5em}
\centering
\subfigure[CIFAR-10]{
\centering
\includegraphics[width=0.8\textwidth]{Poisoned_Samples_CIFAR_new.pdf}
}\vspace{-0.7em}
\subfigure[ImageNet]{
\centering
\includegraphics[width=0.8\textwidth]{Poisoned_Samples_ImageNet_new.pdf}
}
\vspace{-0.7em}
\caption{The example of samples involved in different backdoor watermarks. In the BadNets, blended attack, WaNet, and UBW-P, the labels of poisoned samples are inconsistent with their ground-truth ones. In the label-consistent attack, Sleeper Agent, and UBW-C, the labels of poisoned samples are the same as their ground-truth ones. In particular, the label-consistent attack can only poison samples in the target class, while other methods can modify all samples. }
\label{fig:poisoned_samples_appendix}
\vspace{-0.7em}
\end{figure}

\section{Detailed Experimental Settings}
\subsection{Detailed Settings for Dataset Watermarking}
\label{sec:settings_watermark}

\subsec{Datasets and Models. }
In this paper, we conduct experiments on two classical benchmark datasets, including CIFAR-10 \cite{krizhevsky2009learning} and (a subset of) ImageNet \cite{deng2009imagenet}, with ResNet-18 \cite{he2016deep}. Specifically, we randomly select a subset containing $50$ classes with $25,000$ images from the original ImageNet for training (500 images per class) and $2,500$ images for testing (50 images per class). For simplicity, all images are resized to $3 \times 64 \times 64$, following the settings used in Tiny-ImageNet \cite{chrabaszcz2017downsampled}.

\subsec{Baseline Selection. }
We compare our UBW with representative existing poison-only backdoor attacks. Specifically, for attacks with poisoned labels, we adopt BadNets \cite{gu2019badnets}, blended attack (dubbed as `Blended') \cite{chen2017targeted}, and WaNet \cite{nguyen2021wanet} as the baseline methods. They are the representative of visible attacks, patch-based invisible attacks, and non-patch-based invisible attacks, respectively. We use the label-consistent attack (dubbed as `Label-Consistent') \cite{turner2019label} and Sleeper Agent \cite{souri2022sleeper} as the representative of attacks with clean labels. Besides, we also include the models trained on the benign dataset (dubbed as `No Attack') as another baseline for reference.

\subsec{Attack Setup. }
We implement BadNets, blended attack, and label-consistent attack based on the open-sourced Python toolbox---\texttt{BackdoorBox} \cite{li2022backdoorbox}. The experiments of Sleeper Agent are conducted based on its official open-sourced codes\footnote{\url{https://github.com/hsouri/Sleeper-Agent}}. We set the poisoning rate $\gamma = 0.1$ for all attacks on both datasets. In particular, since the label-consistent attack can only modify samples from the target class, its poisoning rate is set to its maximum ($i.e.$, 0.02) on the ImageNet dataset. The target label $y_t$ is set to 1 for all targeted attacks. Besides, following the classical settings in existing papers, we adopt a white-black square as the trigger pattern for BadNets, blended attack, label-consistent attack, and UBW-P on both datasets. The trigger patterns adopted for training Sleeper Agent and UBW-C are sample-specific, while those used in the inference process are the same as those used by BadNets, blended attack, label-consistent attack, and UBW-P. Specifically, for the blended attack, the blended ratio $\alpha$ is set to 0.1; For the label-consistent attack, we used the projected gradient descent (PGD) \cite{madry2018towards} to generate adversarial perturbations within the $\ell^{\infty}$-ball for pre-processing selected images before the poisoning, where the maximum perturbation size $\epsilon=16$, step size 1.5, and 30 steps. For the WaNet, we adopted its default settings provided by \texttt{BackdoorBox} with noise mode. For both Sleeper Agent and our UBW-C, we alternatively optimize the upper-level and lower-level sub-problems 5 times, where we train the model 50 epochs and generate the trigger patterns with PGA-40 on the CIFAR-10 dataset. On the ImageNet dataset, we alternatively optimize the upper-level and lower-level sub-problems 3 times, where we train the model 40 epochs and generate the trigger patterns via PGA-30. The initial model parameters are obtained by training on the benign dataset. We set $\lambda=2$ and the source class is set as 0 on both datasets. The example of poisoned training samples generated by different attacks is shown in Figure \ref{fig:poisoned_samples_appendix}.

\subsec{Training Setup. }
On both CIFAR-10 and ImageNet datasets, we train the model 200 epochs with batch size 128. Specifically, we use the SGD optimizer with a momentum of 0.9, weight decay of $5 \times 10^{-4}$, and an initial learning rate of 0.1. The learning rate is decreased by a factor of 10 at the epoch of 150 and 180, respectively. In particular, we add trigger patterns before performing the data augmentation with horizontal flipping.

\begin{figure}[!t]
\centering
\subfigure[]{
\centering
\includegraphics[width=0.2\textwidth]{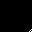}
}\hspace{1.5em}
\subfigure[]{
\centering
\includegraphics[width=0.2\textwidth]{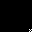}
}\hspace{1.5em}
\subfigure[]{
\centering
\includegraphics[width=0.2\textwidth]{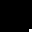}
}
\vspace{-0.7em}
\caption{The trigger patterns used for evaluation.}
\label{fig:patterns}
\end{figure}

\subsection{Detailed Settings for Dataset Ownership Verification}
\label{sec:settings_verification}

We evaluate our verification method in three representative scenarios, including \textbf{1)} independent trigger (dubbed as `Independent-T'), \textbf{2)} independent model (dubbed as `Independent-M'), and \textbf{3)} unauthorized dataset usage (dubbed as `Malicious'). In the first scenario, we query the attacked suspicious model using the trigger that is different from the one used for model training; In the second scenario, we examine the benign suspicious model using the trigger pattern; We adopt the trigger used in the training process of the watermarked suspicious model in the last scenario. Moreover, we sample $m=100$ samples on CIFAR10 and $m=30$ samples on ImageNet and set $\tau = 0.25$ for the hypothesis-test in each case for both UBW-P and UBW-C. We use $m=30$ on ImageNet since there is only 50 testing images from the source class and we only select samples that can be correctly classified by the suspicious model to reduce the side-effects of model accuracy.

\subsection{Detailed Settings for Resistance to Backdoor Defenses}

\subsec{Settings for Fine-tuning. }
We conduct the experiments on the CIFAR-10 dataset as an example for discussion. Following its default settings, we freeze the convolutional layers and tune the remaining fully-connected layers of the watermarked DNNs. Specifically, we adopt 10\% benign training samples for fine-tuning and set the learning rate as 0.1. We fine-tune the model 100 epochs in total.

\subsec{Settings for Model Pruning. }
We conduct the experiments on the CIFAR-10 dataset as an example for discussion. Following its default settings, we conduct channel pruning \cite{he2017channel} on the output of the last convolutional layer with 10\% benign training samples. The pruning rate $\beta \in \{0\%, 2\%, \cdots, 98\%\}$.

\section{The Effects of Trigger Patterns and Sizes}
\subsection{The Effects of Trigger Patterns}

\subsec{Settings. }
In this section, we conduct experiments on the CIFAR-10 dataset to discuss the effects of trigger patterns. Except for the trigger pattern, all other settings are the same as those used in Section \ref{sec:settings_watermark}. The adopted trigger patterns are shown in Figure \ref{fig:patterns}.

\vspace{0.2em}
\subsec{Results. }
As shown in Table \ref{tab: effects_pattern}, both UBW-P and UBW-C are effective with each trigger pattern, although the performance may have some fluctuations. Specifically, the ASR-As are larger than 80\% in all cases. These results verify that both UBW-P and UBW-C can reach promising performance with arbitrary user-specified trigger patterns used in the inference process.

\begin{table}[!t]
\centering
\caption{The effectiveness of our UBW with different trigger patterns on the CIFAR-10 dataset.}
\vspace{-0.5em}
\begin{tabular}{c|c|ccc|c}
\toprule
Method$\downarrow$                 & Pattern$\downarrow$, Metric$\rightarrow$     & BA (\%) & ASR-A (\%) & ASR-C (\%) & $D_p$ \\ \hline
\multirow{3}{*}{UBW-P} & Pattern (a) &   90.59 &   92.30    &    92.51   &     2.2548 \\
 & Pattern (b) &  90.31  &    84.53   &   82.39    &   2.2331   \\
& Pattern (c) &  90.21  &    87.78   &    86.94   &   2.2611   \\ \hline
\multirow{3}{*}{UBW-C} & Pattern (a) &  86.99 & 89.80  &  87.56   &   1.2641 \\
   & Pattern (b) & 86.25   &  90.90   &     88.91  &   1.1131   \\
  & Pattern (c) &  87.78  &    81.23   &   78.55    &  1.0089    \\ \bottomrule
\end{tabular}
\label{tab: effects_pattern}
\end{table}

\begin{table}[!t]
\centering
\caption{The effectiveness of our UBW with different trigger sizes on the CIFAR-10 dataset.}
\vspace{-0.5em}

\begin{tabular}{c|c|ccc|c}
\toprule
\multicolumn{1}{c|}{Method$\downarrow$}  & Trigger Size$\downarrow$, Metric$\rightarrow$   & BA (\%) & ASR-A (\%) & ASR-C (\%) & $D_p$ \\ \hline
\multirow{6}{*}{UBW-P} & 2  &  90.55  &   82.60  &  82.21  &  2.2370    \\
 & 4  &  90.37 & 83.50  & 83.30 & 2.2321   \\
 & 6  & 90.43  &  86.30  &  86.70   &  2.2546    \\ 
 & 8 & 90.46 & 86.40  & 86.26  &  2.2688     \\
 & 10  &  90.72  & 86.10 &  85.82   & 2.2761    \\ 
 & 12  & 90.22  &  88.30  &  87.94  & 2.2545     \\ \hline
\multirow{6}{*}{UBW-C} & 2  &  87.34  &   4.38  &  15.00    &   0.7065    \\
 & 4  &  87.71  &   70.80  &  64.86    &   1.2924    \\
 & 6  &  87.69  &   75.60  &  70.85    &   1.7892    \\ 
 & 8 & 88.89  &  75.40  & 69.86  &   1.2904    \\
 & 10  &  88.30  &  77.60  &  73.92   &  1.7534   \\
 & 12  &  89.29  &  98.00  &   97.72  & 1.1049     \\
\bottomrule
\end{tabular}
\label{tab: effects_size}
\end{table}

\begin{table}[ht]
\centering
\caption{The p-value of UBW-based dataset ownership verification $w.r.t.$ the verification certainty $\tau$ on the CIFAR-10 dataset.}
\vspace{-0.5em}
\begin{tabular}{c|c|cccccc}
\toprule
    Method$\downarrow$  & Scenario$\downarrow$, $\tau$$\rightarrow$               & 0                   & 0.05                   & 0.1                   & 0.15                   & 0.2 & 0.25 \\ \hline
\multirow{3}{*}{UBW-P} 
& Independent-T &    0.1705    & 1.0 &   1.0   &  1.0   &   1.0  &   1.0  \\
  & Independent-M &     0.2178   &   1.0     &  1.0   &   1.0   &   1.0  &  1.0   \\
 & Malicious     &          $10^{-51}$             &           $10^{-48}$            &          $10^{-45}$             &     $10^{-42}$                  &  $10^{-39}$    &  $10^{-36}$    \\ \hline
\multirow{3}{*}{UBW-C} 
& Independent-T &  $10^{-8}$    &  $10^{-5}$   &   0.0049  & 0.1313  &  0.6473  &  0.9688  \\
 & Independent-M & 0.1821 & 0.9835 & 1.0 & 1.0 &  1.0   & 1.0  \\
 & Malicious  & $10^{-27}$  & $10^{-24}$  & $10^{-22}$  & $10^{-19}$  &  $10^{-16}$    & $10^{-14}$     \\ \bottomrule
\end{tabular}
\label{tab:effects_certainty}
\end{table}

\subsection{The Effects of Trigger Sizes}
\subsec{Settings. }
In this section, we conduct experiments on the CIFAR-10 dataset to discuss the effects of trigger sizes. Except for the trigger size, all other settings are the same as those used in Section \ref{sec:settings_watermark}. The specific trigger patterns are generated based on resizing the one used in our main experiments.

\vspace{0.2em}
\subsec{Results. }
As shown in Table \ref{tab: effects_size}, the attack success rate increases with the increase of trigger size. In particular, different from existing (targeted) patch-based backdoor attacks ($e.g.$, BadNets and blended attack), increasing the trigger size has minor adverse effects in reducing the benign accuracy, which is most probably due to our untargeted attack paradigm. The benign accuracy even slightly increases with the increase of trigger sizes on UBW-C, which is mostly because the trigger pattern is not directly added to the poisoned samples during the training process (as described in Section \ref{sec:opt_details}).

\section{The Effects of Verification Certainty and Number of Sampled Images}
\subsection{The Effects of Verification Certainty}

\subsec{Settings. }
In this section, we conduct experiments on the CIFAR-10 dataset to discuss the effects of verification certainty $\tau$ in UBW-based dataset ownership verification. Except for the $\tau$, all other settings are the same as those used in Section \ref{sec:settings_verification}.

\vspace{0.2em}
\subsec{Results. }
As shown in Table \ref{tab:effects_certainty}, the p-value increases with the increase of verification certainty $\tau$ in all scenarios. In particular, when $\tau$ is smaller than 0.15, UBW-C will misjudge the cases of Independent-T. This failure is due to the untargeted nature of our UBW and why we introduced $\tau$ in our verification process. Besides, the larger the $\tau$, the unlikely the misjudgments happen and the more likely that the dataset stealing is ignored. People should assign $\tau$ based on their specific needs.

\subsection{The Effects of the Number of Sampled Images}

\subsec{Settings. }
In this section, we conduct experiments on the CIFAR-10 dataset to study the number of sampled images $m$ in UBW-based dataset ownership verification. Except for the $m$, all other settings are the same as those used in Section \ref{sec:settings_verification}.

\begin{table}[ht]
\centering
\caption{The p-value of UBW-based dataset ownership verification $w.r.t.$ the number of sampled images $m$ on the CIFAR-10 dataset.}
\vspace{-0.5em}
\begin{tabular}{c|c|cccccc}
\toprule
    Method$\downarrow$  & Scenario$\downarrow$, $m$$\rightarrow$               & 20                   & 40                   & 60                   & 80                   & 100 & 120 \\ \hline
\multirow{3}{*}{UBW-P} & Independent-T &           1.0  &     1.0                 &       1.0               &        1.0              &  1.0   &   1.0  \\
& Independent-M &     1.0                 &        1.0              &          1.0            &    1.0                  &  1.0   &  1.0   \\
& Malicious     &    $10^{-7}$                  &        $10^{-14}$               &        $10^{-23}$               &    $10^{-32}$                   &  $10^{-36}$    &  $10^{-42}$    \\ \hline
\multirow{3}{*}{UBW-C} & 
Independent-T &  0.9348 &   0.9219 & 0.9075  &  0.9093  &  0.9688   & 0.9770\\
&Independent-M & 1.0 & 1.0 & 1.0 & 1.0 &   1.0  & 1.0    \\
  & Malicious  & $10^{-3}$   & $10^{-6}$   &  $10^{-7}$   &  $10^{-10}$    &   $10^{-14}$    &   $10^{-16}$    \\ \bottomrule
\end{tabular}
\label{tab:effects_m}
\end{table}

\begin{table}[!ht]
\centering
\caption{The effectiveness of UBW-C when attacking all samples or samples from the source class.}
\vspace{-0.5em}
\begin{tabular}{c|c|ccc|c}
\toprule
   Dataset$\downarrow$                       & Scenario$\downarrow$, Metric$\rightarrow$       & BA (\%) & ASR-A (\%) & ASR-C (\%) & $D_p$ \\ \hline
\multirow{2}{*}{CIFAR-10} & All &    \textbf{87.42}     &  58.83    &   50.31    &  0.9843   \\
 & Source (Ours)    &  86.99       &     \textbf{89.80}       &    \textbf{87.56}        &   \textbf{1.2641}      \\ \hline
\multirow{2}{*}{ImageNet} & All &    58.64    &     42.03      &   21.27     &  2.1407    \\
  & Source (Ours)   &   \textbf{59.64}      &     \textbf{74.00}       &      \textbf{60.00}      &    \textbf{2.4010}    \\ \bottomrule
\end{tabular}
\label{tab:effects_source}
\end{table}

\vspace{0.2em}
\subsec{Results. }
As shown in Table \ref{tab:effects_m}, the p-value decreases with the increase of $m$ in the malicious scenario while it decreases with the increase of $m$ in the independent scenarios. In other words, the probability that our UBW-based dataset ownership verification makes correct judgments increases with the increase of $m$. This benefit is mostly because increasing $m$ will reduce the adverse effects of the randomness involved in the sample selection.

\section{The Effectiveness of UBW-C When Attacking All Classes}\label{sec:effects_all}
As described in Section \ref{sec:opt_details}, our UBW-C randomly selects samples from a random source class instead of all classes for gradient matching. This special design is to reduce the optimization difficulty, since the gradient ensemble of samples from different classes may be too `noisy' to learn for the poisoned training image generator $G_t$. In this section, we verify its effectiveness by comparing our UBW-C with its variant, which uses all samples for gradient matching.

As shown in Table \ref{tab:effects_source}, only using source class samples is significantly better than using all samples during the optimization of UBW-C. Specifically, the ASR-A increases of UBW-C compared with its variant are larger than 30\% on both CIFAR-10 and ImageNet. Besides, we notice that the averaged prediction dispersibility $D_p$ of the UBW-C variant is similar to that of our UBW-C to some extent. It is mostly because our UBW-C is untargeted and the variant has relatively low benign accuracy.

\section{Resistance to Other Backdoor Defenses}
In this section, we discuss the resistance of our UBW-P and UBW-C to more potential backdoor defenses. We conduct experiments on the CIFAR-10 dataset as an example for the discussion.

\subsection{Resistance to Trigger Synthesis based Defenses}
Currently, there are many trigger synthesis based backdoor defenses \cite{wang2019neural,guo2022aeva,tao2022better}, which synthesized the trigger pattern for backdoor unlearning or detection. Specifically, they first generate the potential trigger pattern for each class and then filter the final synthetic one based on anomaly detection. In this section, we verify that our UBW can bypass these defenses for it breaks their latent assumption that the backdoor attacks are targeted.

\subsec{Settings. }
Since neural cleanse \cite{wang2019neural} is the first and the most representative trigger synthesis based defense, we adopt it as an example to synthesize the trigger pattern of DNNs watermarked by BadNets and our UBW-P and UBW-C. We implement it based on its open-sourced codes\footnote{\url{https://github.com/bolunwang/backdoor}} and default settings.

\subsec{Results. }
As shown in Figure \ref{fig:NC}, the synthesized pattern of BadNets is similar to the ground-truth trigger pattern. However, those of our UBW-P and UBW-C are significantly different from the ground-truth one. These results show that our UBW is resistant to trigger synthesis based defenses.

\begin{figure}[!t]
\centering
\subfigure[]{
\centering
\includegraphics[width=0.22\textwidth]{trigger_real.png}
}\hspace{0.3em}
\subfigure[]{
\centering
\includegraphics[width=0.22\textwidth]{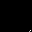}
}\hspace{0.3em}
\subfigure[]{
\centering
\includegraphics[width=0.22\textwidth]{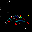}
}\hspace{0.3em}
\subfigure[]{
\centering
\includegraphics[width=0.22\textwidth]{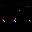}
}
\caption{The ground-truth trigger pattern and those synthesized by neural cleanse. \textbf{(a)}: ground-truth trigger pattern; \textbf{(b)}: synthesized trigger pattern of BadNets; \textbf{(c)}: synthesized trigger pattern of UBW-P; \textbf{(d)}: synthesized trigger pattern of UBW-C. The synthesized pattern of BadNets is similar to the ground-truth one whereas those of our UBW-P and UBW-C are meaningless.}
\label{fig:NC}
\end{figure}

\begin{figure}[!t]
    \centering
    \includegraphics[width=0.95\textwidth]{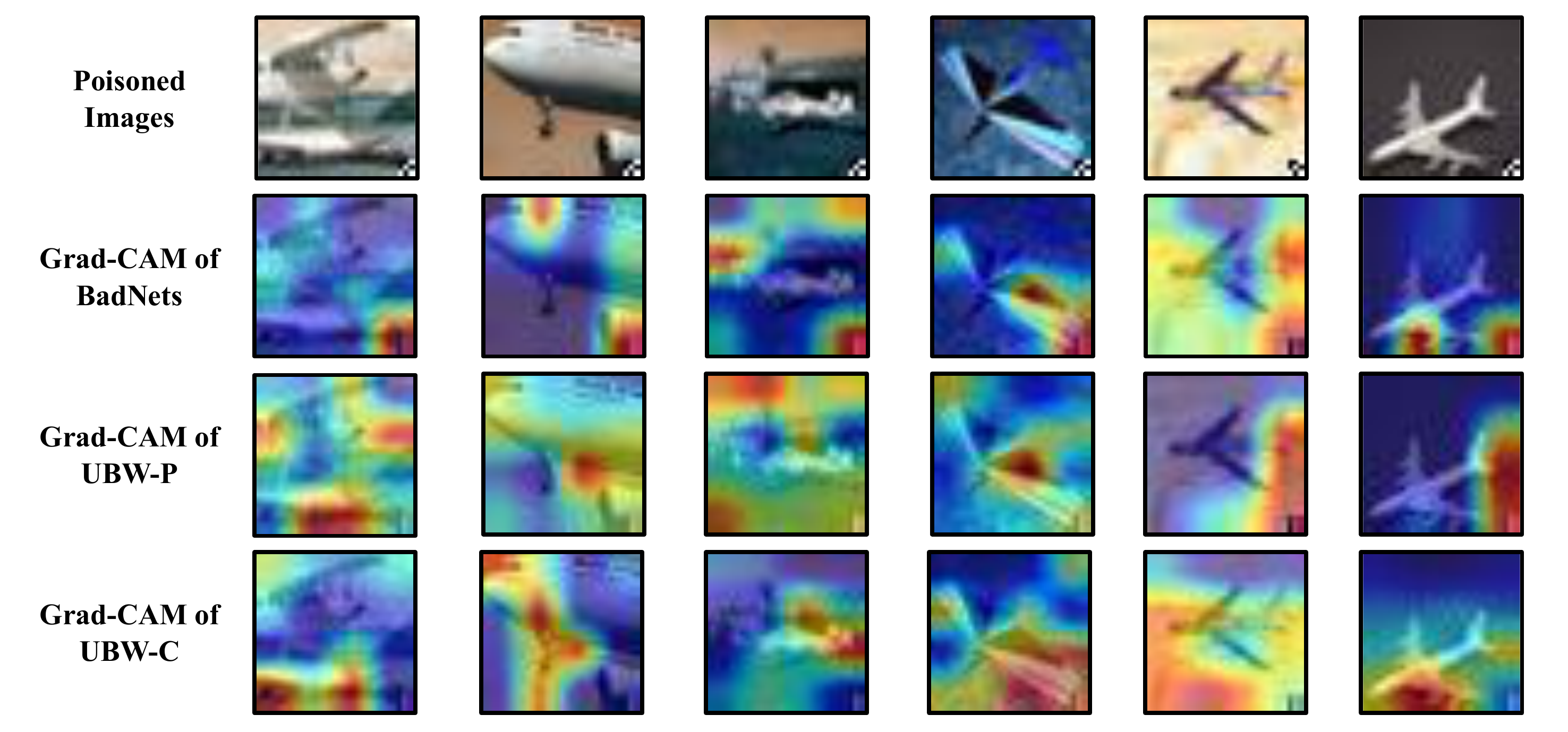}
    \caption{The poisoned images and their saliency maps based on Grad-CAM with DNNs watermarked by different methods. The Grad-CAM mainly focuses on the trigger areas of poisoned images in BadNets, while it mainly focuses on other regions ($e.g.$, object outline) in our UBW.}
    \label{fig:saliency}
\end{figure}

\subsection{Resistance to Saliency-based Defenses}
Since the attack effectiveness is mostly caused by the trigger pattern, there were also some backdoor defenses \cite{huang2019neuroninspect,chou2020sentinet,doan2020februus} based on detecting trigger areas with saliency maps. Specifically, these methods first generated the saliency map of each sample and then obtained trigger regions based on the intersection of all generated saliency maps. Since our UBW is untargeted, the relation between the trigger pattern and the predicted label is less significant compared with existing targeted backdoor attacks. As such, it can bypass those saliency-based defenses, which is verified in this section.

\subsec{Settings. }
We generate the saliency maps of models watermarked by BadNets and our UBW-P and UBW-C, based on the Grad-CAM \cite{selvaraju2017grad} with its default settings. We randomly select samples from the source class to generate their poisoned version for the discussion.

\subsec{Results. }
As shown in Figure \ref{fig:saliency}, the Grad-CAM mainly focuses on the trigger areas of poisoned images in BadNets. In contrast, it mainly focuses on other regions ($e.g.$, object outline) of poisoned images in our UBW-C. We notice that the Grad-CAM also focuses on the trigger areas in our UBW-P in a few cases. It is most probably because the trigger pattern used in the inference process is the same as the one used for training in our UBW-P while we use invisible additive noises in our training process of UBW-C. These results validate that our UBW is resistant to saliency-based defenses.

\begin{table}[!t]
\centering
\caption{The averaged entropy generated by STRIP of models watermarked by different methods. The larger the entropy, the harder for STRIP to detect the watermark.}
\vspace{-0.5em}
\begin{tabular}{c|ccc}
\toprule
Metric$\downarrow$, Method$\rightarrow$ & BadNets & UBW-P  & UBW-C  \\ \midrule
Averaged Entropy        & 0.0093  & 1.5417 & 1.2018 \\ \bottomrule
\end{tabular}
\label{tab:strip}
\end{table}

\begin{table}[!t]
\centering
\caption{The successful filtering rate (\%) on the CIFAR-10 dataset.}
\vspace{-0.5em}
\begin{tabular}{c|cc}
\toprule
Method$\downarrow$, Defense$\rightarrow$ & Spectral Signatures & Activation Clustering \\ \hline
UBW-P           & 10.96               & 52.61                 \\
UBW-C           & 9.40                & 20.51                 \\ \bottomrule
\end{tabular}
\label{tab:resistance_SS_AC}
\end{table}

\subsection{Resistance to STRIP}
Recently, Gao \etal \cite{gao2021design} proposed STRIP to filter poisoned samples based on the prediction variation of samples generated by imposing various image patterns on the suspicious image. The variation is measured by the entropy of the average prediction of those samples. Specifically, the STRIP assumed that the trigger pattern is sample-agnostic and the attack is targeted. Accordingly, the more likely the suspicious image contains trigger pattern, the smaller the entropy since those modified images will still be predicted as the target label so that the average prediction is still nearly an one-hot vector.

\subsec{Settings. }
We randomly select 100 testing images from the source class to generate their poisoned version, based on BadNets and our UBW-P and UBW-C. We calculate the entropy of each poisoned image based on the open-sourced codes\footnote{\url{https://github.com/yjkim721/STRIP-ViTA}} and default settings of STRIP. We then calculate the averaged entropy among all poisoned samples for each watermarking method as their indicator. The larger the entropy, the harder for STRIP to detect the watermark.

\subsec{Results. }
As shown in Table \ref{tab:strip}, the averaged entropies of both UBW-P and UBW-C are significantly higher than that of BadNets. Specifically, the entropies of both UBW-P and UBW-C are more than 100 times larger than that of BadNets. It is mostly due to the untargeted nature of our UBW whose predictions are dispersed. These results verify that our UBW is resistant to STRIP.

\subsection{Resistance to Dataset-level Backdoor Defenses}
In this section, we discuss whether our methods are resistant to dataset-level backdoor defenses.

\subsec{Settings. }
In this part, we adopt the spectral signatures \cite{tran2018spectral} and the activation clustering \cite{chen2019detecting} as representative dataset-level backdoor defenses for our discussion. Both spectral signatures and activation clustering tend to filter poisoned samples from the training dataset, based on sample behaviors in hidden feature space. We implement these methods based on their official open-sourced codes with default settings on the CIFAR-10 dataset. Besides, we adopt the \emph{successful filtering rate} defined as the number of filtered poisoned samples over that of all filtered samples as our evaluation metric. In general, the smaller the successful filtering rate, the more resistance of our UBW.

\subsec{Results. }
As shown in Table \ref{tab:resistance_SS_AC}, these defenses fail to filter our watermarked samples under both poisoned-label and clean-label to some extent. We speculate that it is mostly because poisoned samples generated by our UBW-P and UBW-C tend to scatter in the whole space instead of forming a single cluster in the feature space. We will further explore it in the future.

\subsection{Resistance to MCR and NAD}
Here we discuss whether our methods are resistant to mode connectivity repairing (MCR) \cite{zhao2020bridging} and neural attention distillation (NAD) \cite{li2021neural}, which are two advanced repairing-based backdoor defenses.

\subsec{Settings. }
We implement MCR and NAD based on the codes provided in \texttt{BackdoorBox}.

\begin{table}[!t]
\centering
\caption{The resistance to MCR and NAD on the CIFAR-10 dataset.}
\vspace{-0.5em}
\begin{tabular}{c|cc|cc|cc}
\toprule
Defense$\rightarrow$        & \multicolumn{2}{c|}{No Defense} & \multicolumn{2}{c|}{MCR} & \multicolumn{2}{c}{NAD} \\ \hline
Method$\downarrow$, Metric$\rightarrow$ & BA (\%)             & ASR-A (\%)          & BA (\%)          & ASR-A (\%)      & BA (\%)         & ASR-A (\%)      \\ \hline
UBW-P          & 90.59          & 92.30          & 88.17       & 96.20      & 67.98      & 99.40      \\
UBW-C          & 86.99          & 89.80          & 86.15       & 79.10      & 77.13      & 36.00      \\ \bottomrule
\end{tabular}
\label{tab:resistance_MCR_NAD}
\end{table}

\begin{table}[!t]
\centering
\caption{The effectiveness of our UBW-P with different types of triggers on the CIFAR-10 dataset.}
\vspace{-0.5em}
\begin{tabular}{c|ccc|c}
\toprule
Method$\downarrow$, Metric$\rightarrow$  & BA (\%)    & ASR-A (\%) & ASR-C (\%) & $D_p$   \\ \hline
UBW-P (BadNets) & 90.59 & 92.30 & 92.51 & 2.2548 \\
UBW-P (WaNet)   & 89.90 & 73.00 & 70.45 & 2.0368 \\ \bottomrule
\end{tabular}
\label{tab:UBW-P_WaNet}
\end{table}

\subsec{Results. }
As shown in Table \ref{tab:resistance_MCR_NAD}, both our UBW-P and UBW-C are resistant to MCR and NAD to some extent. Their failures are probably because both of them contain a fine-tuning stage, which is ineffective for our UBWs (as demonstrated in Section \ref{sec:resistance_defenses}).

\section{UBW-P with Imperceptible Trigger Patterns}
In our main manuscript, we design our UBW-P based on BadNets-type triggers since it is the most straightforward method. We intend to show how simple it is to design UBW under the poisoned-label setting. Here we demonstrate that our UBW-P is still effective with imperceptible trigger patterns.

\subsec{Settings. }
We adopt the advanced invisible targeted backdoor attack -- WaNet \cite{nguyen2021wanet} to design our UBW-P with imperceptible trigger patterns. We also implement it based on the open-sourced codes of vanilla WaNet provided in \texttt{BackdoorBox} \cite{li2022backdoorbox}. Specifically, we set the warping kernel size as 16 and conduct experiments on the CIFAR-10 dataset. Except for the trigger patterns, all other settings are the same as those used in our standard UBW-P.

\subsec{Results. }
As shown in Table \ref{tab:UBW-P_WaNet},  our UBW-P can still reach promising performance with imperceptible trigger patterns, although it may have relatively low ASR compared to UBW-P with the BadNets-type visible trigger. It seems that there is a trade-off between ASR and trigger visibility. We will discuss how to better balance the watermark effectiveness and its stealthiness in our future work.

\section{The Transferability of our UBW-C}
Recall that in the optimization process of our UBW-C, we need to know the model structure $f$ in advance. Following the classical settings of bi-level-optimization-type backdoor attacks ($e.g.$, LIRA \cite{doan2021lira} and Sleeper Agent \cite{souri2022sleeper}), we report the results of attacking DNN with the same model structure as the one used for generating poisoned samples. In practice, dataset users may adopt different model structures since dataset owners have no information about the model training. In this section, we evaluate whether the watermarked dataset is still effective in watermarking DNNs having different structures compared to the one used for dataset generation ($i.e.$, transferability).

\subsec{Settings. }
We adopt ResNet-18 to generate a UBW-C training dataset, based on which to train different models ($i.e.$, ResNet-18, ResNet-34, VGG-16-BN, and VGG-19-BN). Except for the model structure, all other settings are the same as those used in Section \ref{sec:exps}.

\begin{table}[!t]
\centering
\caption{The performance of our UBW-C with different model structures trained on the watermarked CIFAR-10 dataset generated with ResNet-18.}
\vspace{-0.5em}
\begin{tabular}{c|cccc}
\toprule
Metric$\downarrow$, Model$\rightarrow$ & ResNet-18 & ResNet-34 & VGG-16-BN & VGG-19-BN \\ \hline
BA (\%)           & 86.99     & 87.34     & 86.83     & 88.55     \\
ASR-A (\%)         & 87.56     & 78.89     & 75.80     & 74.30     \\ \bottomrule
\end{tabular}
\label{sec:transferability}
\end{table}

\subsec{Results. }
As shown in Table \ref{sec:transferability}, our UBW-C has high transferability. Accordingly, our methods are practical in protecting open-sourced datasets.

\section{Connections and Differences with Related Works}
In this section, we discuss the connections and differences between our UBW and adversarial attacks, data poisoning, and classical untargeted attacks. We also discuss the connections and differences between our UBW-based dataset ownership verification and model ownership verification.

\subsection{Connections and Differences with Adversarial Attacks}
Both our UBW and adversarial attacks intend to make the DNNs misclassify samples during the inference process by adding malicious perturbations. However, they still have some intrinsic differences. 

Firstly, the success of adversarial attacks is mostly due to the behavior differences between DNNs and humans, while that of our UBW results from the data-driven training paradigm and excessive learning ability of DNNs. Secondly, the malicious perturbations are known ($i.e.$, non-optimized) by UBW whereas adversarial attacks need to obtain them based on the optimization process. As such, adversarial attacks cannot to be real-time in many cases, since the optimization requires querying the DNNs multiple times under either white-box \cite{bai2020targeted,croce2020reliable,chen2022adversarial} or black-box \cite{chen2017zoo,andriushchenko2020square,feng2022boosting} settings. Lastly, our UBW requires modifying the training samples without any additional requirements in the inference process, while adversarial attacks need to control the inference process to some extent.

\subsection{Connections and Differences with Data Poisoning}
Currently, there are two types of data poisoning, including classical data poisoning \cite{xiao2015feature,koh2017understanding,feng2019learning} and advanced data poisoning \cite{shafahi2018poison,geiping2021witches,schwarzschild2021just}. Specifically, the former type of data poisoning intends to reduce model generalization, so that the attacked DNNs behave well on training samples whereas having limited effectiveness in predicting testing samples. The latter one requires that the model has good benign accuracy while misclassifying some adversary-specified unmodified samples.

Our UBW shares some similarities to data poisoning in the training process. Specifically, they all intend to embed distinctive prediction behaviors in the DNNs by poisoning some training samples. However, they also have many essential differences. The detailed differences are as follows:

\textbf{The Differences Compared with Classical Data Poisoning. }
Firstly, UBW has a different goal compared with classical data poisoning. Specifically, UBW preserves the accuracy in predicting benign testing samples whereas classical data poisoning is not. Secondly, UBW is also more stealthy compared with classical data poisoning, since dataset users can easily detect classical data poisoning by evaluating model performance on a local verification set. In contrast, this method has limited benefits in detecting UBW. Lastly, the effectiveness of classical data poisoning is mostly due to the sensitiveness of the training process, so that even a small domain shift of training samples may lead to significantly different decision surfaces of attacked models. It is different from that of our UBW.

\textbf{The Differences Compared with Advanced Data Poisoning. }
Firstly, advanced data poisoning can only misclassify a few selected images whereas UBW can lead to the misjudgments of all images containing the trigger pattern. It is mostly due to their second difference that data poisoning does not require modifying the (benign) images before the inference process. Thirdly, the effectiveness of advanced data poisoning is mainly because DNNs are over-parameterized, so that the decision surface can have sophisticated structures near the adversary-specified samples for misclassification. It is also different from that of our UBW.

\subsection{Connections and Differences with Classical Untargeted Attacks}

Both our UBW and classical untargeted attacks ($e.g.$, untargeted adversarial attacks) intend to make the model misclassify specific sample(s). However, different from existing classical untargeted attacks which simply maximize the loss between the predictions of those samples and their ground-truth labels, our UBW also requires optimizing the prediction dispersibility so that the adversaries cannot deterministically manipulate model predictions. Maximizing only the untargeted loss may not be able to disperse model predictions, since targeted attacks can also maximize that loss when the target label is different from the ground-truth one of the sample. Besides, introducing prediction dispersibility may also increase the difficulty of the untargeted attack since it may contradict the untargeted loss to some extent (as described in Section \ref{sec:opt_details}).

\subsection{Connections and Differences with Model Ownership Verification}
Our UBW-based dataset ownership verification enjoys some similarities to model ownership verification \cite{adi2018turning,jia2021entangled,li2022defending} since they all conduct verification based on the distinctive behaviors of DNNs. However, they still have many fundamental differences, as follows:

Firstly, dataset ownership verification has different threat models and requires different capacities. Specifically, model ownership verification is adopted to protect the copyrights of open-sourced or deployed models, while our method is for protecting dataset copyrights. Accordingly, our UBW-based method only needs to modify the dataset, whereas model ownership verification usually also requires controlling other training components ($e.g.$, loss). In other words, our UBW-based method can also be exploited to protect model copyrights, whereas most of the existing methods for model ownership verification are not capable to protect (open-sourced) datasets.

Secondly, to the best of our knowledge, almost all existing black-box model ownership verification was designed based on the targeted attacks ($e.g.$, targeted poison-only backdoor attacks) and therefore introducing new security risks in DNNs. In contrast, our verification method is mostly harmless, since our UBW used for dataset watermarking is untargeted and with high prediction dispersibility.

\subsection{Connections and Differences with Radioactive Data}
We notice that radioactive data (RD) \cite{sablayrolles2020radioactive} (under the black-box setting) can also be exploited as dataset watermarking for ownership verification by analyzing the loss of watermarked and benign images. If the loss of watermarked images is significantly lower than that of their benign version, RD treats the suspicious model as trained on the protected dataset. Both RD and UBW-C require knowing the model structure in advance, although they all have transferability. However, they still have many fundamental differences, as follows:

Firstly, our UBWs have a different verification mechanism compared to RD. Specifically, UBWs adopt the change of predicted probability on the ground-truth label, while RD exploits the loss change for verification. In practice, it is relatively difficult to select the confidence budget for RD since the loss values may change significantly across different datasets. In contrast, users can easily select the confidence budget ($i.e.$, $\tau$) from $[0, 1]$ since the predicted probability on the ground-truth label are relatively stable ($e.g.$, nearly 1 for benign samples).

Secondly, our UBWs require fewer defender capacities compared to RD. RD needs to have the prediction vectors or even the model source files for ownership verification, whereas UBWs only require the probability in the predicted label. Accordingly, our method can even be generalized to the scenario that users can only obtain the predicted labels (as suggested in \cite{li2022black}), based on examining whether poisoned images have different predictions compared to their benign version, whereas RD cannot. We will further discuss the label-only UBW verification in our future work.

Lastly, it seems that RD is far less effective on datasets with relatively low image resolution and fewer samples ($e.g.$, CIFAR-10)\footnote{\url{https://github.com/facebookresearch/radioactive_data/issues/3}}. In contrast, our methods have promising performance on them.

\section{Discussions about Adopted Data}
In this paper, all adopted samples are from the open-sourced datasets ($i.e.$, CIFAR-10 and ImageNet). The ImageNet dataset may contain a few personal contents, such as human faces. However, our research treats all objects the same and does not intentionally exploit or manipulate these contents. Accordingly, our work fulfills the requirements of those datasets and should not be regarded as a violation of personal privacy. Besides, our samples contain no offensive content, since we only add some invisible noises or non-semantic patches to a few benign images.

\end{document}